\pdfoutput=1
\documentclass[10pt,a4paper,reqno]{amsart}
\usepackage{amssymb,amsmath}
\usepackage{algorithm}
\usepackage{algorithmic}
\usepackage{hyperref}
\usepackage{enumerate}
\usepackage{booktabs}
\usepackage{subfigure}
\usepackage{mathtools}
\usepackage{etoolbox}

\theoremstyle{plain}%     

\newtheorem{theorem}{Theorem}[]
\newtheorem{lemma}{Lemma}[]
\newtheorem{corollary}{Corollary}[]

\theoremstyle{definition}%   

\newcommand{\1}{{\textrm{1} \kern -.41em \textrm{1} }}

\numberwithin{equation}{section}

\newcommand{\vja}{v_j^{(1)}}
\newcommand{\vjb}{v_j^{(2)}}
\newcommand{\vjc}{v_j^{(3)}}

\begin{document}
\title{On double-resolution imaging and discrete tomography}
\author{Andreas Alpers and Peter Gritzmann}
\address{Zentrum Mathematik, Technische Universit\"at M\"unchen, D-85747 Garching bei M\"unchen, Germany}
\email{alpers@ma.tum.de, gritzman@tum.de}
%\thanks{The authors gratefully acknowledge support through the German Research Foundation Grant GR 993/10-2 and the European COST Network MP1207.}

%\date{\today, \textsc{File:} \texttt{\jobname.pdf}}

\maketitle
\begin{abstract}
Super-resolution imaging aims at improving the resolution of an image by enhancing it with other images or data that might have been acquired using different imaging techniques or modalities. In this paper we consider the task of doubling, in each dimension, the resolution of grayscale images of binary objects by fusion  with double-resolution tomographic data that have been acquired from two viewing angles. We show that this task is polynomial-time solvable if the gray levels have been reliably determined. The problem becomes $\mathbb{N}\mathbb{P}$-hard if the gray levels of some pixels come with an error of $\pm1$ or larger. The  $\mathbb{N}\mathbb{P}$-hardness persists for any larger resolution enhancement factor. This means that noise does not only affect the quality of a reconstructed image but, less expectedly, also the algorithmic tractability of the inverse problem itself.
\end{abstract}

% REQUIRED
%\begin{keywords}
%discrete mathematics, combinatorics, discrete tomography, super-resolution, \\
%polynomial-time algorithms, computational complexity 
%\end{keywords}

% REQUIRED
%\begin{AMS}
% 68R05, 68Q25, 49N45, 94A08, 68Q25
%\end{AMS}

\section{Introduction} 
Different imaging techniques in tomography have different characteristics that strongly depend on the specific data acquisition setup and the imaged tissue/material. Hence it is a major issue (and at the heart of current research, see, e.g., \cite{modal6, modal4, breaklimit1, modal7, modal3, modal5, modal2}) to improve the resolution of an image by combining different imaging techniques. In general it is, however, not clear, how this can actually be performed efficiently. The present paper addresses this multimodal approach from a basic algorithmic point of view in two different ways using a discrete model of binary image reconstruction. First, we show that, indeed, the resolution of binary objects can be efficiently improved if the data are reliable. We derive a polynomial-time algorithm for double-resolution that works correctly for exact data (and can be extended to a heuristic for the more general case). On the other hand, we show that noisy data in the low-resolution image do not simply reduce the quality of the reconstruction but add additional complexity. Hence, intuitively speaking, and in a sense that will be made mathematically precise later, it might not be possible to compensate for the faultiness of low-resolution imaging efficiently by incorporating other additional higher-resolution information (even if the latter is noise-free).

More specifically, the results of the present paper are motivated by the task of enhancing the resolution of reconstructed tomographic images obtained from binary objects representing, for instance, crystalline structures, nanoparticles or two-phase samples \cite{batenburgnature, alpersgardner13, agms-15, DTnature2,  glidingarc-15}. The reconstructed tomographic image might contain several gray levels, which, depending on the accuracy of the reconstruction, result from the fact that low-resolution pixels may cover different numbers of black high-resolution pixels. For turning the grayscale image into a high-resolution binary image we utilize the gray levels and two additional high-resolution projections, which may have been acquired by different imaging techniques or modalities (e.g., via scanning transmission electron microscopy \cite{DTnature2}). More precisely, we study the task of reconstructing binary~$m\times n$-images from row and column sums and additional constraints, so-called \emph{block constraints}, on the number of black pixels to be contained in the $k\times k$-blocks resulting from a subdivision of each pixel in the $m/k\times n/k$ low-resolution image.  We remark that we do not require that the X-ray data are taken from orthogonal directions. In our context it suffices that the X-ray data have been taken with high-resolution according to the discretization of the low-resolution image.

Figure~\ref{fig:process} illustrates the process. The example given in this figure also shows that the block-constraints can help to narrow down the solution space. In fact, the solution shown in this figure is uniquely determined by the input (the row and column sums and block constraints). The row and column sums alone do not determine the solution uniquely and, of course, neither do the block constraints.

\begin{figure}[htb]
\centering
\includegraphics[width=0.37\textwidth]{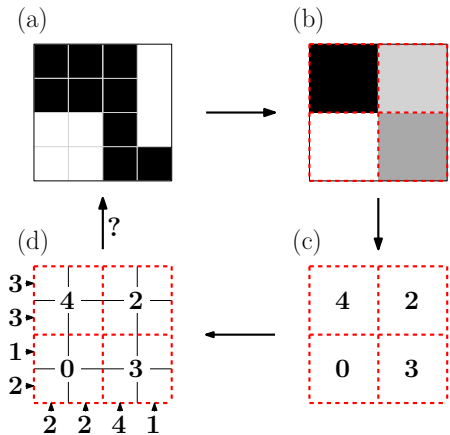}
\caption{The super-resolution imaging task \textsc{DR}. (a) Original (unknown) high-resolution image, (b) the corresponding low-resolution grayscale image, (c) gray levels converted into block constraints, (d) taken in combination with high-resolution row and column sum data. The task is to reconstruct from~(d) the original binary image shown in~(a). } \label{fig:process}
\end{figure}

Apart from super-resolution imaging \cite{superresolutionCT2, siam, superresolutionCT3, superresolutionbook, superresolution1, superresolutionCT4}, and its particular applications \cite{apkh-06, superresolution2, rpkh07, superresolutionCT1} in discrete tomography \cite{batenburgnature, gg97, gritzmann97, kubaherman1, kubaherman2, irvingjerrum94, ksbsko-95, kahkrp-09, sksbko-93}, the problems discussed in this paper are also relevant in other contexts. 

From a combinatorial point of view, they can be viewed as reconstructing binary matrices from given row and column sums and some additional constraints.  Unexpected complexity jumps based on the results of the present paper are discussed in~\cite{agwindowconstraints} and, to some extent, at the end of this section. Background information on such problems involving different kinds of additional constraints can be found in \cite[Sect.~4]{brualdi}. 

Other applications belong to the realm of \emph{dynamic discrete tomography}~\cite{agdynamic, agwindowconstraints, agms-15, glidingarc-15}. For instance, in plasma particle tracking, some particles are reconstructed at time~$t.$ Between~$t$ and the next time step~$t+1$ the particles may have moved to other positions. The task is to reconstruct their new positions, again from few projections. One way of incorporating additional prior knowledge about the movement leads to block constraints of the kind discussed in this paper; see \cite{agdynamic}.  

As the task of reconstructing a binary matrix from its row and column sums can be formulated as the task of finding a $b$-matching in a bipartite graph, we remark that our theorems can also be seen as results on finding $b$-matchings that are subject to specific additional constraints. Related (but intrinsically quite different) results on matchings subject to so-called budget constraints can be found, e.g., in \cite{budget2}.

For super-resolution imaging the two main contributions of this paper are as follows. On the one hand, we show that reliable bimodal tomographic data can be utilized efficiently, i.e., the resolution, in each dimension, can be doubled in polynomial time if the gray levels of the low-resolution image have been determined precisely (Thm.~\ref{thm:main1}). On the other hand, the task becomes already intractable (unless $\mathbb{P}=\mathbb{N}\mathbb{P}$) if the gray levels of some pixels come with some small error of $\pm1$ (Thm.~\ref{thm:main2} and Cor.~\ref{cor:largebox}). This proves that noise does not only affect the quality of the reconstructed image but also the algorithmic tractability of the inverse problem itself. Hence the possibility of compensating noisy imaging by including bimodal information may in practice be jeopardized by its algorithmic complexity. While this is not the focus of the present paper let us point out that, even in the presence of (a reasonably restricted degree of) noise, the approach leading to Thm.~\ref{thm:main1} can, in principle, be extended to a fast heuristic for doubling the resolution of such images in each dimension. At present, however, we do not know how such an approach would compare with other heuristics on real-world data.

From the perspective of discrete tomography, the contributions of this paper can be interpreted as follows. It is well known that the problem of reconstructing binary images from X-ray data taken from two directions can be solved in polynomial time~\cite{gale57, ryser57}. Typically, this information does not determine the image uniquely (see, e.g., \cite{glw11} and the papers quoted there). Hence, one would like to take and utilize additional measurements. If, however, we add additional constraints that enforce that the solutions satisfy the X-ray data taken from a third direction, then the problem becomes $\mathbb{N}\mathbb{P}$-hard, and it remains $\mathbb{N}\mathbb{P}$-hard if X-ray data from even more directions are given \cite{ggp-99} (see also \cite{duerr-Guinez-matamala-12} for results on a polyatomic version). 

As it turns out, the case of block constraints behaves somewhat differently. Thm.~\ref{thm:main2} and Cor.~\ref{cor:largebox} show that the problem of reconstructing a binary image from X-ray data taken from two directions is again $\mathbb{N}\mathbb{P}$-hard if we add \emph{several} (but not all) block constraints (which need to be satisfied with equality).  
However, and possibly less expectedly, if we include \emph{all} block constraints, then the problem becomes polynomial-time solvable (Thm.~\ref{thm:main1}). If, on the other hand, from \emph{all} block constraints \emph{some} of the data come with \emph{noise} at most~$\pm1,$ then the problem becomes again $\mathbb{N}\mathbb{P}$-hard (Thm.~\ref{thm:main2} and Cor.~\ref{cor:largebox}). And yet again, if from \emph{all} block constraints \emph{all} of the data are \emph{sufficiently noisy}, then the problem is in~$\mathbb{P}$ (as this is again the problem of reconstructing binary images from X-ray data taken from two directions). An overview of these complexity jumps is given in Fig.~\ref{fig:complexityjumps}.

\begin{figure}[htb]
\centering
\includegraphics[width=0.5\textwidth]{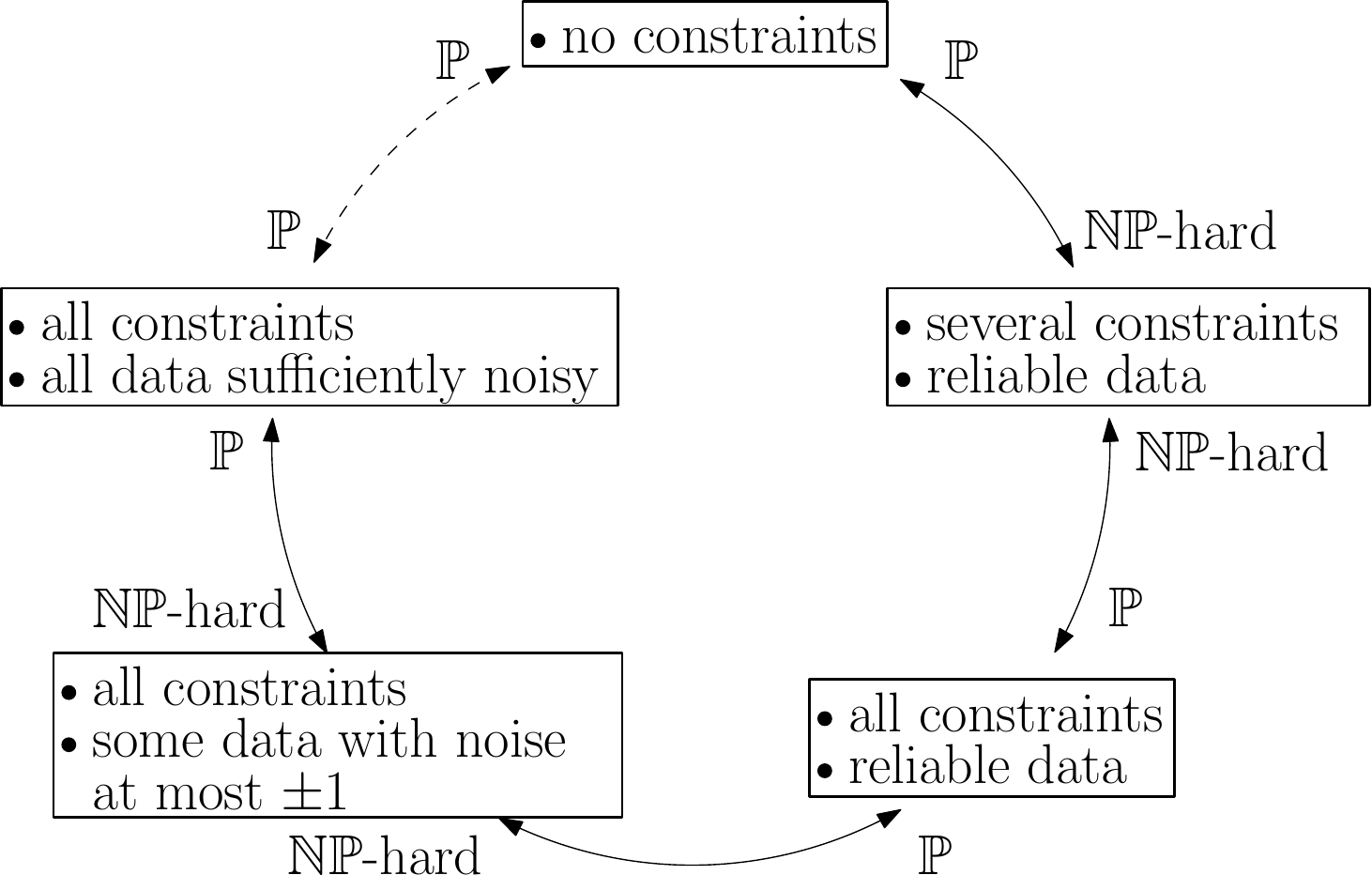}
\caption{Overview of complexity jumps for the problem of reconstructing a binary image from row and column sums and additional $2\times 2$-block constraints.} \label{fig:complexityjumps}
\end{figure}

\section{Notation and Main Results}\label{sect:notation}

Let $\mathbb{Z},$ $\mathbb{N},$ and $\mathbb{N}_0$ denote the set of integers, natural numbers, and non-negative integers, respectively. For $k \in \mathbb{N}$ set $k\mathbb{N}_0:=\{ki:i\in \mathbb{N}_0\},$ $[k]:=\{1,\dots,k\},$ and $[k]_0:=\{0,\dots,k\}.$ 
%The support of a vector $x:=(\xi_1,\dots,\xi_d) \in \mathbb{Z}^d$ is defined as  $\textnormal{supp}(x):=\{i:\xi_i\neq0\}$. 
With $\1$ we denote the all-ones vector of the corresponding dimension. The cardinality of a finite set $F\subseteq \mathbb{Z}^d$ is denoted by~$|F|.$

In this paper we use Cartesian coordinates (rather than matrix notation) to represent pixels in an image. In particular, the two numbers $i$ and $j$ in a pair $(i,j)$ denote the $x$- and $y$-coordinate of a point, respectively. In the following we consider grids~$[m]\times[n].$
The set $\{i\}\times[n]$ and $[m]\times\{j\}$ is called \emph{column~$i$} and \emph{row~$j,$} respectively.  Let, in the following, $k\in \mathbb{N}.$ Any set $\{(i-1)k+1,\dots,ik\}\times[n]$ is a \emph{vertical strip (of width~$k$)}. A \emph{horizontal strip (of width~$k$)} is a set of the form $[m]\times\{(j-1)k+1,\dots,jk\}.$

Sets of the form $([a,b]\times[c,d])\cap\mathbb{Z}^2,$ with $a,b,c,d\in\mathbb{Z}$ and $a\leq b,$ $c\leq d,$ are called \emph{boxes.}  For $i,j\in\mathbb{N}$ let $B_k(i,j):=B(i,j):=(i,j)+[k-1]_0^2.$  Defining for any $k\in\mathbb{N}$ and $m,n\in k\mathbb{N}$ the set of \emph{(lower-left) corner points} $C(m,n,k):=([m]\times [n])\cap(k\mathbb{N}_0+1)^2,$ we call any box $B_k(i,j)$ with $(i,j)\in C(m,n,k)$ a \emph{block}. The blocks form a partition of $[m]\times[n],$ i.e., $\dot{\bigcup}_{(i,j)\in C(m,n,k)}B_k(i,j)=[m]\times[n].$  

Let us remark that, as defined, all our images consist of lattice points. However, the figures will be given in terms of pixels. Accordingly, the block $B_k(i,j)$ will be depicted as $(i,j)+[0,k-1]^2;$ see Fig.~\ref{fig:gridimage}. The individual pixels are $(p,q)+[0,1]^2,$ $p,q\in[2],$ i.e., are identified with the lattice point at their lower left corner. Of course, the fact that a point is or is not present in a solution is indicated by a black or a white pixel, respectively.

\begin{figure}[htb]
\centering
\includegraphics[width=0.4\textwidth]{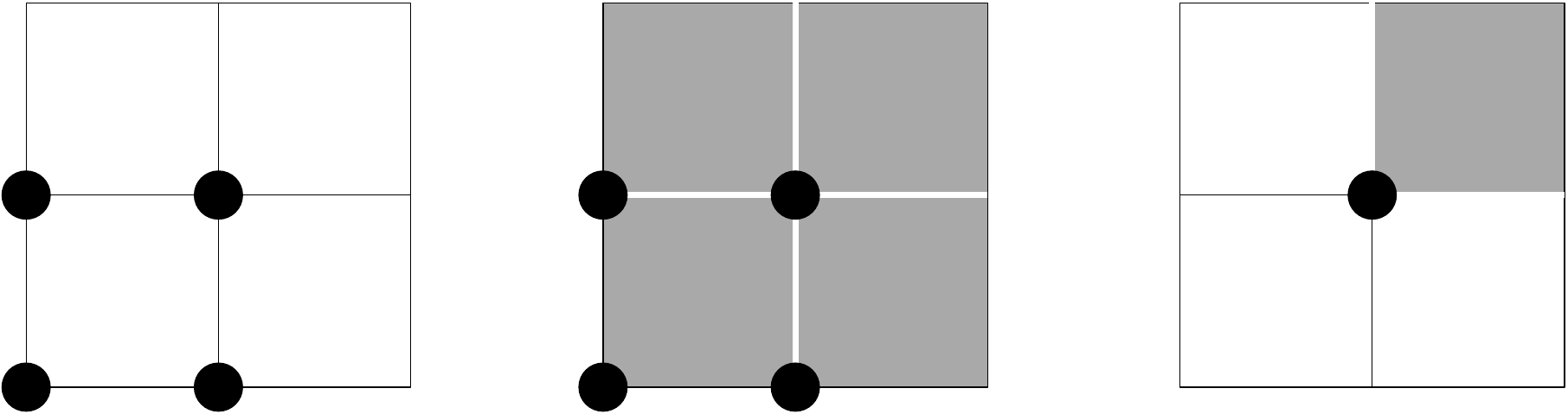}
\caption{Lattice points (left) and pixels (middle); right: pixel associated with its lattice point.}\label{fig:gridimage}
\end{figure}

For $\varepsilon,k\in\mathbb{N}_0$ with $k\geq2$ we define the task of (noisy) super-resolution\\ 

\hspace*{-4ex}
\begin{minipage}{.9\textwidth}
{\small
\textsc{nSR}$(k,\varepsilon)$\\[-5ex]

\hspace*{1ex}\begin{minipage}{0.99\textwidth}
\begin{alignat*}{8}
&\textnormal{Instance:}\:&& \omit \rlap{$\displaystyle m,n \in k\mathbb{N},$}\\
&&&\omit \rlap{$\displaystyle r_1,\dots,r_n \in \mathbb{N}_0,$}&&&&&& \textnormal{(row sum measurements)} \\
&&&\omit \rlap{$\displaystyle c_1,\dots,c_m \in \mathbb{N}_0,$}&&&&&& \textnormal{(column sum measurem.)} \\
&&&\omit \rlap{$\displaystyle R\subseteq C(m,n,k),$}&&&&&& \textnormal{(corner points of reliable}\\[-1ex]
&&&&&&&&& \textnormal{\hspace*{1ex}gray value measurem.)}\\
&&&\omit \rlap{$\displaystyle v(i,j)\in[k^2]_0,$} &&&&(i,j)\in C(m,n,k), &&  \textnormal{(gray value measurem.)}\\[1.2ex]
&\textnormal{Task:}\:&& \omit \rlap{Find  $\xi_{p,q}\in\{0,1\},$ \:$(p,q)\in[m]\times[n],$ with }\\
\displaystyle &&&\sum_{\mathclap{p\in[m]}} \xi_{p,q}&&=r_q,&&  q\in[n], &&\textnormal{(row sums)} \\ 
\displaystyle &&&\sum_{q\in[n]} \xi_{p,q}&&=c_p,&& p \in [m], &&\textnormal{(column sums)}\\  
\displaystyle &&&\sum_{\mathclap{(p,q)\in B_k(i,j)}}\xi_{p,q}&&=v(i,j),&& (i,j)\in R, &&\textnormal{(block constraints)}\\
\displaystyle &&&\sum_{\mathclap{(p,q)\in B_k(i,j)}}\xi_{p,q}&&\in v(i,j)+[-\varepsilon,\varepsilon], &\hspace*{3ex}& (i,j)\in C(m,n,k)\setminus R,\:&&\textnormal{(noisy block constraints)}\\
&&& \omit \rlap{or decide that no such solution exists.}\\
\end{alignat*}
\end{minipage}
}
\end{minipage}

The numbers $r_1,\dots,r_n$ and $c_1,\dots,c_m$ are the row and column sum measurements of the high-resolution binary $m\times n$ image, $v(i,j)\in[k^2]_0$ corresponds to the gray value of the low-resolution $k\times k$-pixel at~$(i,j)$ of the low-resolution $m/k\times n/k$ grayscale image, and $R$ is the set of low-resolution pixel locations for which we assume that the gray values have been determined reliably, i.e., without error. The number $\varepsilon$ is an error bound for the remaining blocks. (While it may seem unusual to denote a non-negative integer by~$\varepsilon$ we chose this notation to indicate the specific role of~$\varepsilon$ of quantifying an error.) The task is to find a binary high-resolution image satisfying the row and column sums such that the number of black pixels in each block is the gray value for the corresponding $k\times k$-pixel low-resolution image. Clearly, a necessary (and easily verified) condition for feasibility is that \[\sum_{q\in[n]}r_q=\sum_{p\in[m]}c_p.\] In the following we will assume without loss of generality that this is always the case.

Our special focus is on double-resolution imaging, i.e., on the case $k=2.$ For~$\varepsilon>0$ we define \textsc{nDR}$(\varepsilon):=$\textsc{nSR}$(2,\varepsilon).$ In the reliable situation, i.e., for $\varepsilon=0,$ we simply speak of \emph{double-resolution} and set  \textsc{DR}$:=$\textsc{nSR}$(2,0).$ (Then, of course, the set~$R$ can be omitted from the input.)

In the present paper, we show that the resolution of any reliable grayscale image can in fact be doubled in each dimension in polynomial-time if X-ray data are provided from two viewing angles at double-resolution. In other words, we show
\begin{theorem}  \label{thm:main1}
$\textsc{DR}\in\mathbb{P}.$ 
\end{theorem}

Figures~\ref{fig:drexample2} and~\ref{fig:drexample1} illustrate the performance of the algorithm for \textsc{DR} applied to two phantoms. The results have been obtained within a fraction of a second on a standard PC.

In Fig.~\ref{fig:drexample2}, the original phantom is a binary $200\times200$ image of a crystalline sample taken from~\cite{philmag}. It is assumed that the low-resolution grayscale image shown in Fig.~\ref{fig:drexample2}(a) has been obtained by some imaging method and that double-resolution X-ray information in the two standard directions is also available. Application of the algorithm (see Algorithm~\ref{alg:alg1}) yields the binary image shown in Fig.~\ref{fig:drexample2}(c).

\begin{figure}[htb]
\centering
\subfigure[]{\includegraphics[width=0.27\textwidth]{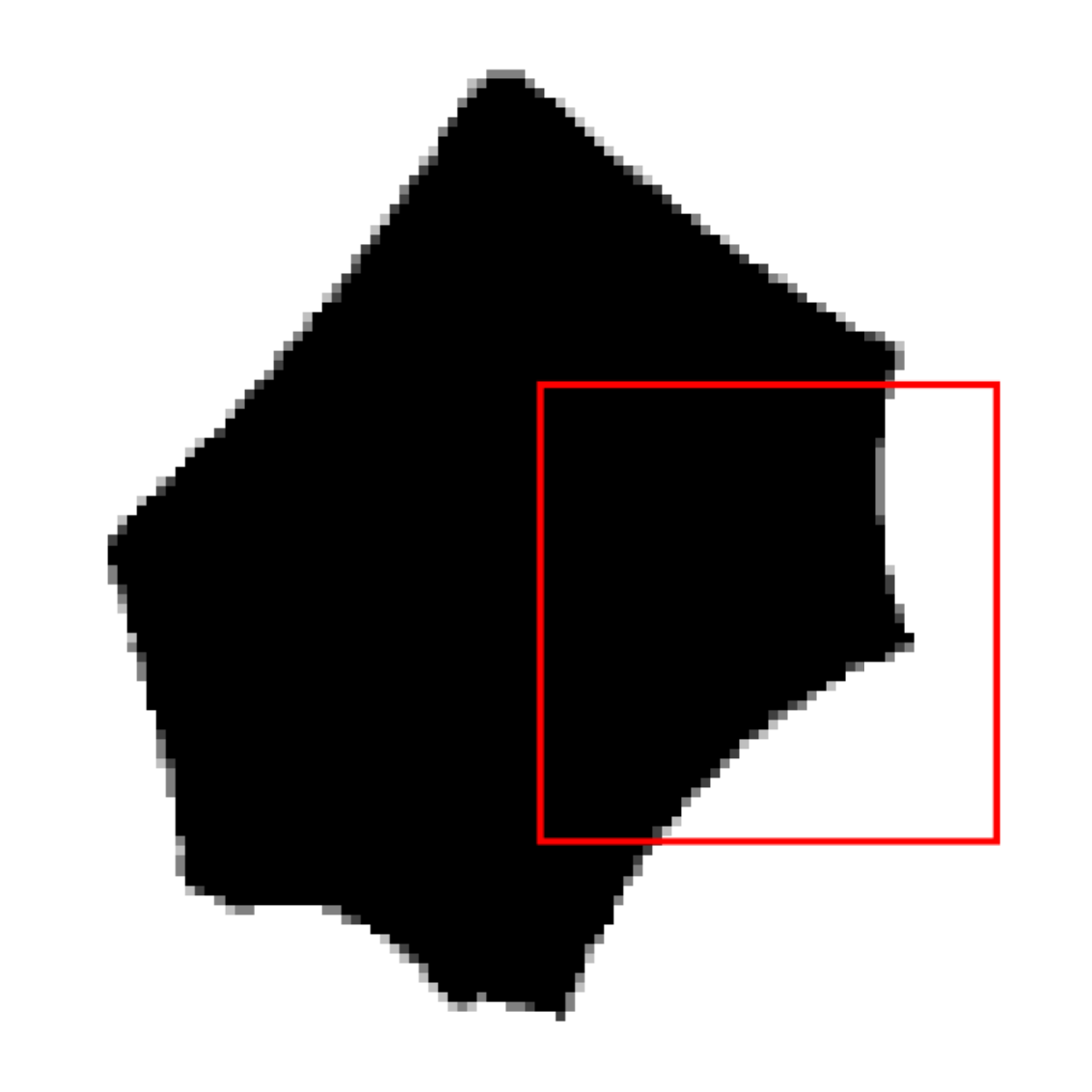}}\hspace*{0ex}
\subfigure[]{\includegraphics[width=0.27\textwidth]{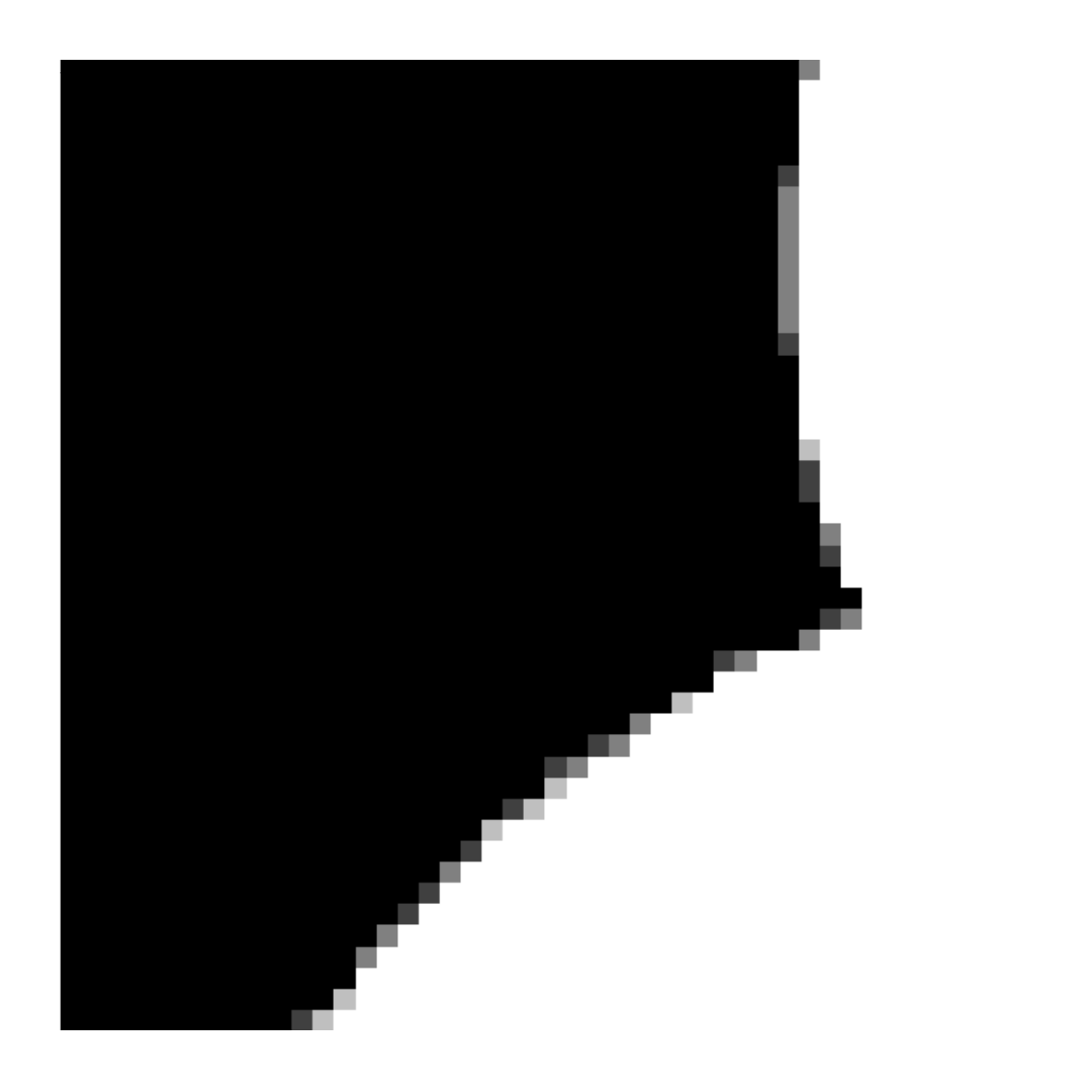}}\hspace*{0ex}
\subfigure[]{\includegraphics[width=0.27\textwidth]{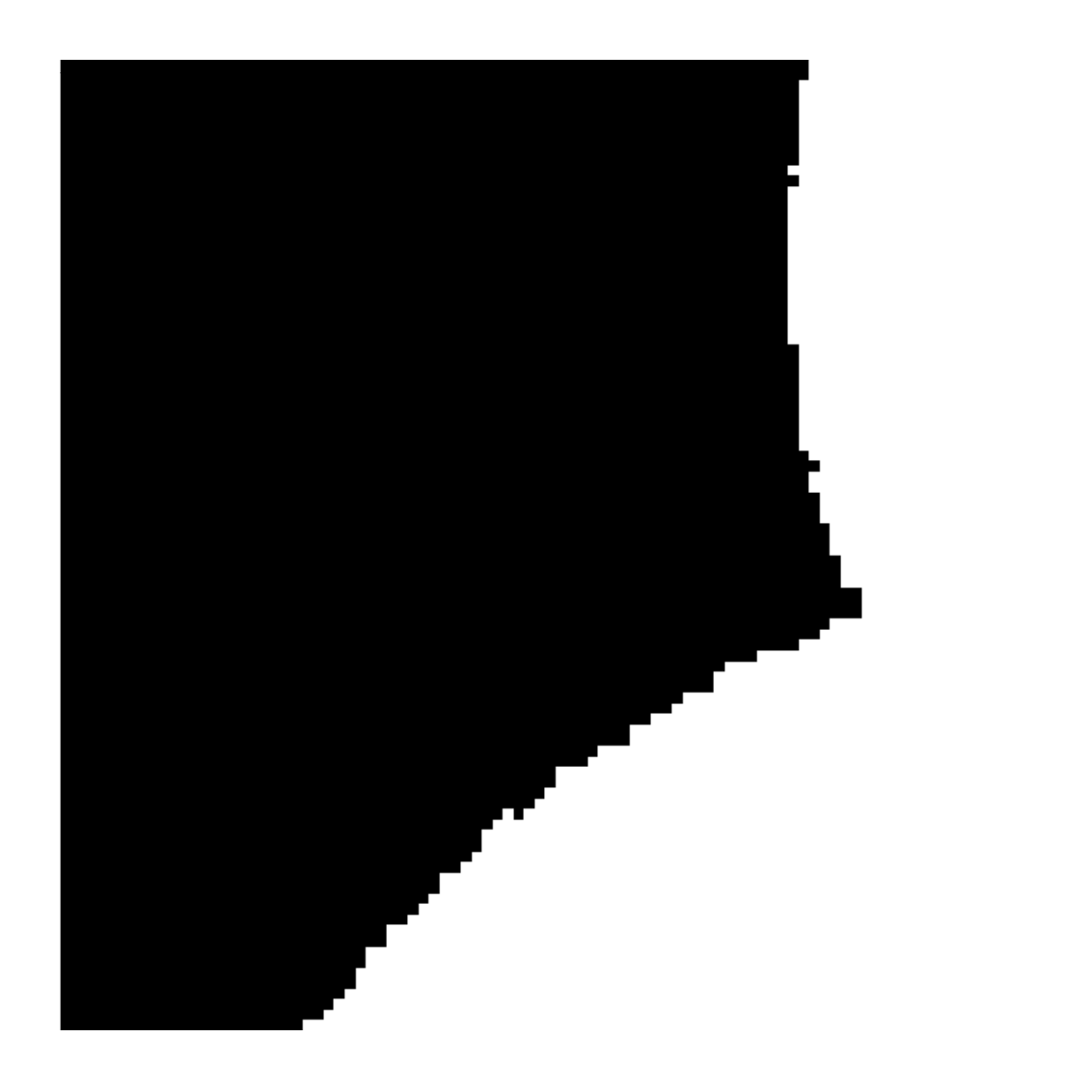}}
\caption{Algorithm~\ref{alg:alg1} applied to a phantom. (a)~Low-resolution $100\times100$ grayscale image, (b)~magnification of the part contained in the red rectangle, (c)~magnification of the corresponding part of the~$200\times200$ high-resolution binary image obtained by the algorithm.} \label{fig:drexample2}
\end{figure}

In Fig.~\ref{fig:drexample1}, the original phantom is a binary $100\times100$ image, for which the low-resolution grayscale image shown in Fig.~\ref{fig:drexample1}(a) is available. 
From this image and the row and column sums (counting the black pixels) of the original phantom the algorithm returns the binary image shown in Fig.~\ref{fig:drexample1}(b). 

\begin{figure}[htb]
\centering
\subfigure[]{\includegraphics[width=0.27\textwidth]{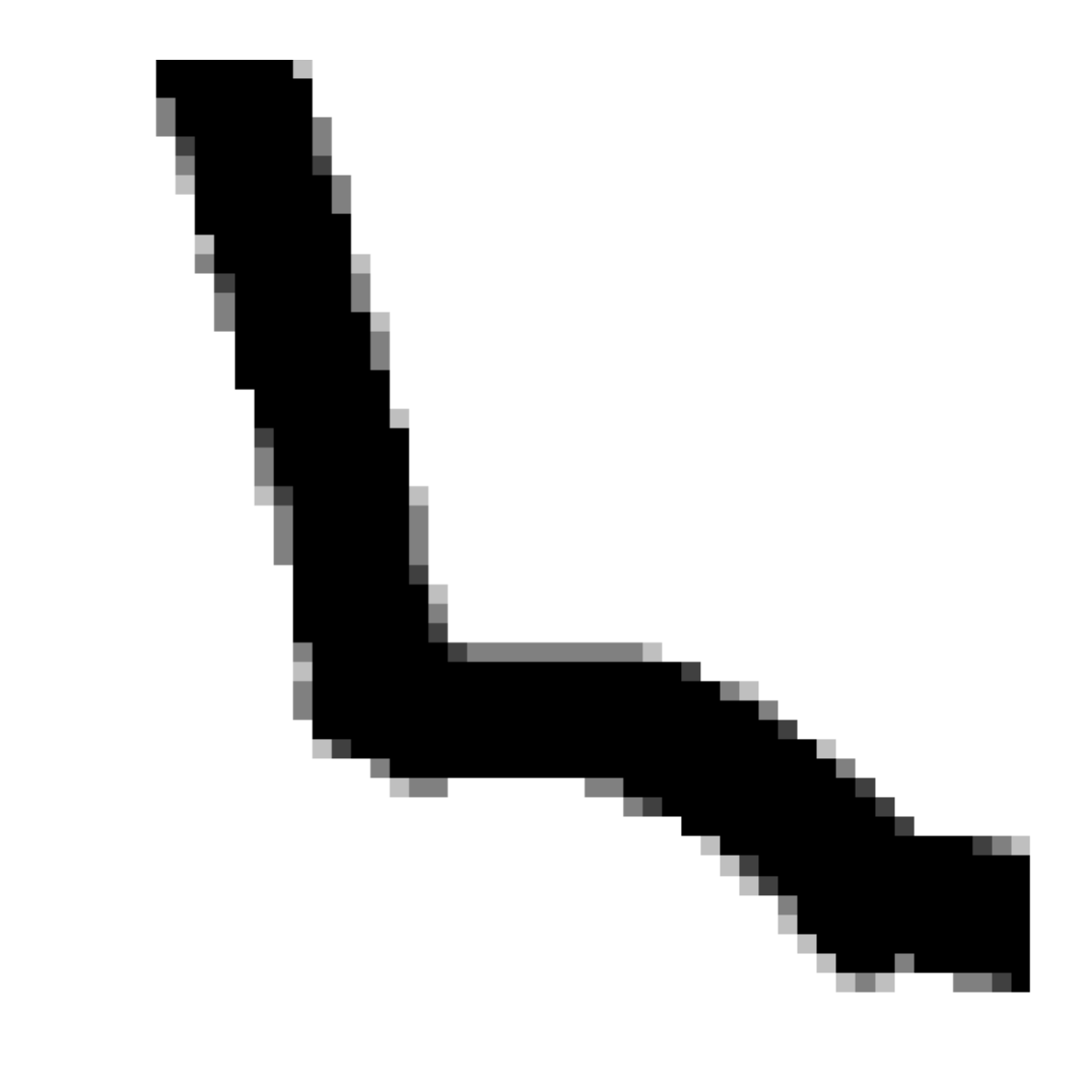}}\hspace*{2ex}
\subfigure[]{\includegraphics[width=0.27\textwidth]{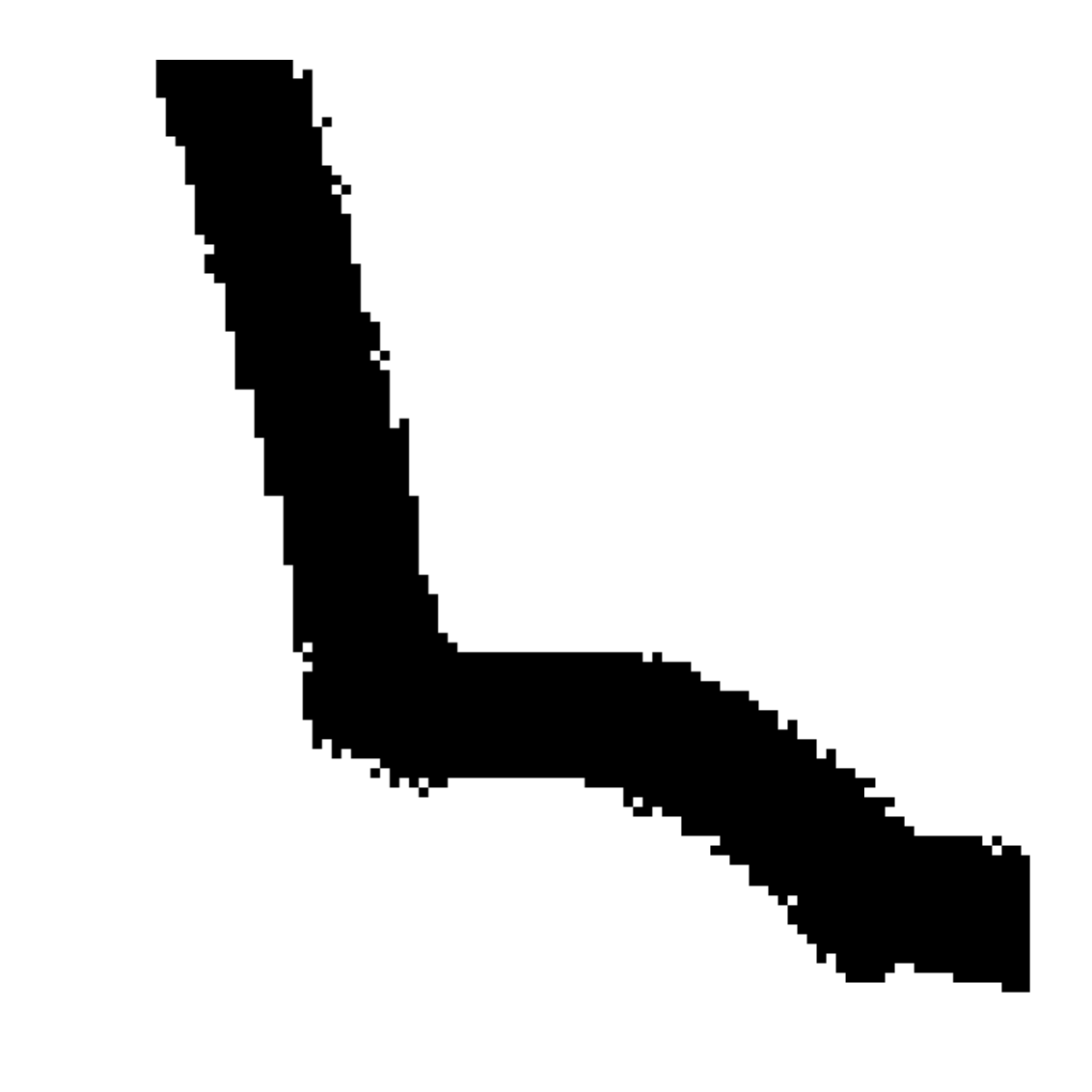}}\hspace*{2ex}
\subfigure[]{\includegraphics[width=0.27\textwidth]{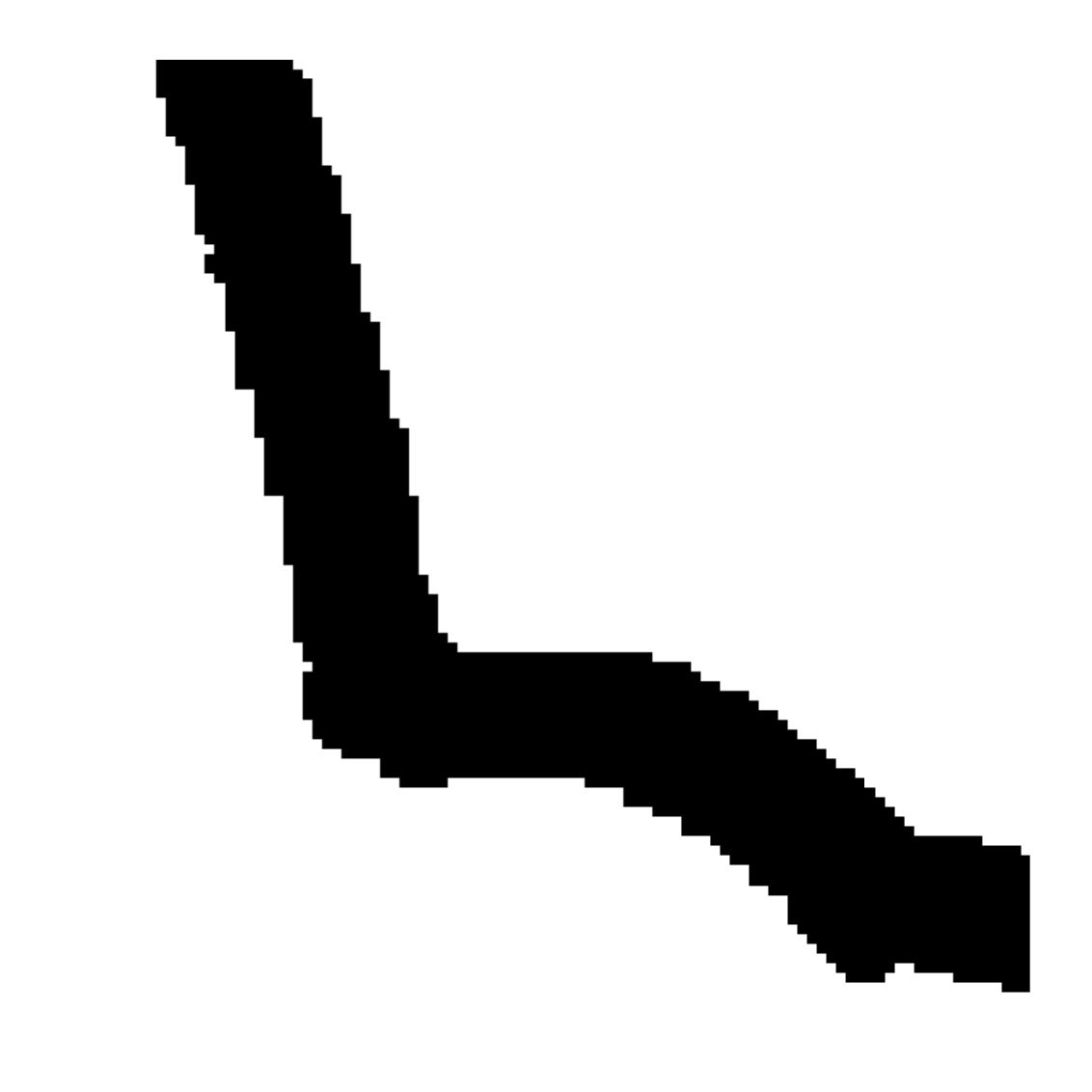}}
\caption{Algorithm~\ref{alg:alg1} applied to a phantom. (a)~Low-resolution $50\times50$ grayscale image,  (b)~high-resolution $100\times100$ binary image obtained by the algorithm, (c)~high-resolution $100\times100$ binary image minimizing the total variation.} \label{fig:drexample1}
\end{figure}

The results shown in Figs.~\ref{fig:drexample2} and~\ref{fig:drexample1} satisfy, of course, the row and column sums and block constraints. However, in both cases the solutions are still not unique. This can also be detected in polynomial time.

\begin{theorem}\label{main:unique}
For every instance of \textsc{DR} it can be decided in polynomial time whether the instance admits a unique solution.
\end{theorem}
 
While the additional X-ray information (at double-resolution) does in general reduce the ambiguity, typically it will still not lead to uniqueness. Of course, a standard way of dealing with non-uniqueness in practice is regularization \cite{apkh-06,bertero, hansen, siltanen}. To illustrate the possibility of adapting a regularization scheme to our context, we show in Fig.~\ref{fig:drexample1}(c) a binary solution minimizing the \emph{total variation (TV).} In fact, the solution is obtained from Fig.~\ref{fig:drexample1}(a) by applying~$15$ so-called local switches, each of which is strictly decreasing the value of the total variation functional 
\[J(x):=\sum_{p\in [m]}\sum_{q\in[n]}||(\nabla x_{p,q}^{(1)},\nabla x_{p,q}^{(2)})^T||_2,\] where
\begin{eqnarray*}
(\nabla x)_{p,q}^{(1)}&:=\left\{\begin{array}{lll}\xi_{p+1,q}-\xi_{p,q} &:& p<m,\\ 0&:&\textnormal{otherwise,}\end{array}\right.\\
(\nabla x)_{p,q}^{(2)}&:=\left\{\begin{array}{lll}\xi_{p,q+1}-\xi_{p,q} &:& q<n,\\ 0&:&\textnormal{otherwise,}\end{array}\right.
\end{eqnarray*} and $||\cdot||_2$ denotes the Euclidean norm. See, e.g., \cite{chambolle} for some background information. Of course, we obtain the same results for other~$p$-norms (including the $1$-norm), because our images are binary.

We now turn to the task $n\textsc{DR}(\varepsilon)$ where small ``occasional'' uncertainties in the gray levels are allowed. First observe that a constant number of uncertainties does not increase the complexity. However, if we allow a (as we will see, comparably small but) non-constant number of uncertainties the problem becomes hard.

\begin{theorem}\label{thm:main2}
\textsc{nDR}$(\varepsilon)$ is $\mathbb{N}\mathbb{P}$-hard for any $\varepsilon>0.$
\end{theorem}

(For background material on complexity theory, see, e.g., \cite{gareyjohnson}.)

Further, the $\mathbb{N}\mathbb{P}$-hardness extends to the task of checking uniqueness.

\begin{corollary} \label{cor:uniqueness}
The problem of deciding whether a given solution of an instance of \textsc{nDR}$(\varepsilon)$ with $\varepsilon>0$ has a non-unique solution is $\mathbb{N}\mathbb{P}$-complete.
\end{corollary}

As it turns out, the problem \textsc{nDR}$(\varepsilon)$ with $\varepsilon>0$ remains $\mathbb{N}\mathbb{P}$-hard for larger block sizes.
\begin{corollary}\label{cor:largebox}
\textsc{nSR}$(k,\varepsilon)$ is $\mathbb{N}\mathbb{P}$-hard for any $k\geq2$ and $\varepsilon>0.$
\end{corollary}
In other words, noise does not only affect the reconstruction quality. It also affects the algorithmic tractability of the inverse problem.

The present paper is organized as follows. We deal with the case of reliable data in section~\ref{sect:3}. Results on \textsc{nDR}$(\varepsilon)$ and \textsc{nSR}$(k,\varepsilon)$ involving noisy data are contained in section~\ref{sect:4}. Section~\ref{sect:5} concludes with some additional remarks on certain extensions and an open problem.

%-----------------------------------------------------------------------------------------------------------

\section{Reliable Data}\label{sect:3}
In this section we discuss the task \textsc{DR}, i.e., $k=2$ and $\varepsilon=0.$ 
%For proving the main results of this section, Theorem~\ref{thm:main1}, we need two lemmas that show polynomial-time solvability of three specific subtasks. 
We remark that, in the following, our emphasis is on providing brief and concise arguments for polynomial-time solvability rather than to focus on computationally or practically most efficient algorithms. Recall, however, our comments on the computational performance of our algorithms and the effect of the additional information for the example images depicted in Figs.~\ref{fig:drexample2} and~\ref{fig:drexample1}.

The main result of this section is~Thm.~\ref{thm:main1}. In its proof we show that \textsc{DR} decomposes into five problems that can be solved independently. The five problems are restricted single-graylevel versions  of \textsc{DR} where each block is required to contain the same number $\nu\in[4]_0$ of ones. In fact, we allow for restricted problem instances (and in particular row and column sums) to be defined only on subsets 
\begin{equation}
G(I):=\bigcup_{(i,j)\in I}B(i,j)\subseteq [m]\times[n] \label{eq:Gdef}
\end{equation} of the grid $[m]\times[n]$ given by means of some $I\subseteq C(m,n,2).$  Let 
\[
\Pi_x(I):=\{i\in[m]:\exists j\in[n]:(i,j)\in I\}
\] and
\[
\Pi_y(I):=\{j \in [n]:\exists i \in [m]:(i,j)\in I\}
\] denote the projection of~$I$ onto the first and second coordinate, respectively. 

Then, we define for $\nu\in[4]_0:$ \\

\begin{center}
\begin{minipage}{0.95\textwidth}
{\small
\textsc{DR}$(\nu)$\\[-3ex]

\hspace*{2ex}\begin{minipage}{0.8\textwidth}
\begin{alignat*}{6}
&\textnormal{Instance:}\quad&& \omit \rlap{$\displaystyle m,n \in 2\mathbb{N},$}\\
&&&\omit \rlap{$\displaystyle I\subseteq C(m,n,2),$}&&\hspace*{15ex}&&&\hspace*{5ex}& \textnormal{(a set of corner points)}\\
&&&\omit \rlap{$\displaystyle r_{j+l}\in \mathbb{N}_0,$}&&&& j \in \Pi_y(I), \:\: l\in \{0,1\}&& \textnormal{(row sum measurem.)}  \\
&&&\omit \rlap{$\displaystyle c_{i+l}\in \mathbb{N}_0,$}&&&& i \in \Pi_x(I), \:\: l\in \{0,1\}&& \textnormal{(column sum measurem.)}\\[1.2ex]
&\textnormal{Task:}\quad&& \omit \rlap{Find  $\xi_{p,q}\in\{0,1\}, \:(p,q)\in G(I)$ with }\\
\displaystyle &&&\sum_{\mathclap{p:(p,j)\in G(I)}}\xi_{p,j+l}&&=r_{j+l},&&  j\in \Pi_y(I), \:\: l\in \{0,1\}&&\textnormal{(row sums)} \\ 
\displaystyle &&&\sum_{\mathclap{q:(i,q)\in G(I)}}\xi_{i+l,q}&&=c_{i+l},&&  i\in \Pi_x(I), \:\: l\in \{0,1\}&&\textnormal{(column sums)}\\  \displaystyle &&&\sum_{\mathclap{(p,q)\in B(i,j)}}\xi_{p,q}&&=\nu,&\:\:& (i,j)\in I, &\:\:&\textnormal{(block constraints)}\\
&&& \omit \rlap{or decide that no such solution exists.}
\end{alignat*}
\end{minipage}
}
\end{minipage}
\end{center}

Of course, the tasks \textsc{DR}$(\nu),$ $\nu\in\{0,4\},$ are trivial. In fact, the only potential solution~$x^*$ satisfying the block constraints is given by $\xi^*_{p,q}=\nu/4,$ $(p,q)\in G(I),$ and it is checked easily and in polynomial time whether this satisfies the row and column sums. 

The next lemma and corollary deal with $\textsc{DR}(\nu)$, $\nu\in\{1,3\}.$ Their statements and proofs use the notation $\sigma_i(j)$ and $\rho_j(i)$ for the number of blocks in the same vertical strip as but below $B(i,j)$ or in the same horizontal strip as but left of $B(i,j),$ respectively. More precisely, let
\[
\sigma_i(j):=|(\{i\}\times[j])\cap I|\qquad \textnormal{and}\qquad \rho_j(i):=|([i]\times\{j\})\cap I|
\] for $(i,j)\in I\subseteq C(m,n,2).$

\begin{lemma} \label{lemma:mono}\hfill
\begin{enumerate}[(i)]
\item \label{mono2} An instance $\mathcal{I}$ of \textsc{DR}$(1)$ is feasible if, and only if, for every~$(i,j) \in I$ we have
\[r_j+r_{j+1}=\rho_j(m) \quad \textnormal{ and } \quad c_i+c_{i+1}=\sigma_i(n).\]
%\[\sum_{l=0}^{k-1}r_{j+l}=\rho_j(m) \quad \textnormal{ and } \quad \sum_{l=0}^{k-1}c_{i+l}=\sigma_i(n);\]
\item \label{mono3} The solution of a feasible instance of \textsc{DR}$(1)$ is unique if, and only if, for every~$(i,j) \in I$ we have 
\[r_j\cdot r_{j+1}=0 \qquad \textnormal{and}\qquad c_i\cdot c_{i+1}=0.\]
\item \label{mono1} \textsc{DR}$(1)\in\mathbb{P}.$ 
\end{enumerate}
\end{lemma}
\begin{proof}
Clearly, for the feasibility of a given instance of \textsc{DR}$(1)$ the conditions 
\[r_j+r_{j+1}=\rho_j(m) \quad \textnormal{ and } \quad c_i+c_{i+1}=\sigma_i(n),\]
%\[\sum_{l=0}^{k-1}r_{j+l}=\rho_j(m) \quad \textnormal{ and } \quad \sum_{l=0}^{k-1}c_{i+l}=\sigma_i(n),\] for every $(i,j)\in I,$
are necessary. 

Now, suppose that the conditions are satisfied. For every $(i,j) \in I$ and $(p,q)\in B(i,j)$ we set
\begin{equation} \label{eq:sol1}
\begin{aligned}
a_{i,j}&:=i+\min\{l\in\{0,1\}:\sigma_{i}(j)\leq \sum_{h=0}^{l}c_{i+h}\},\\
b_{i,j}&:=j+\min\{l\in\{0,1\}:\rho_{j}(i)\leq \sum_{h=0}^{l}r_{j+h}\},\\
\xi^*_{p,q}&:=\left\{\begin{array}{lll}1&:& (p,q)=(a_{i,j},b_{i,j}),\\ 0&:&\textnormal{otherwise.} \end{array}\right.
\end{aligned} 
%\begin{aligned}
%a_i&:=i-1+\min\{l\in[k]:\sigma_{i}(j)\leq \sum_{h=0}^{l-1}c_{i+h}\},\\
%b_j&:=j-1+\min\{l\in[k]:\rho_{j}(i)\leq \sum_{h=0}^{l-1}r_{j+h}\},\\
%\xi^*_{p,q}&:=\left\{\begin{array}{lll}1&:& (p=a_i)\wedge(q=b_j),\\ 0&:&\textnormal{otherwise.} \end{array}\right.
%\end{aligned} 
\end{equation}
Figure~\ref{fig:Example1} gives an illustration. 

\begin{figure}[htb]
\centering
\includegraphics[width=0.45\textwidth]{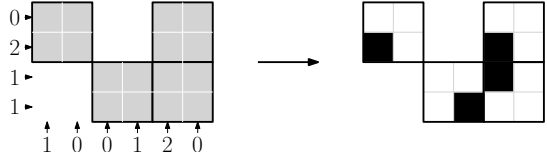}
\caption{Illustration of \textsc{DR}(1). (Left) Row and column sums and blocks $B(i,j)$ with $(i,j)\in I$ in gray color. (Right) Solution defined by~\eqref{eq:sol1}.}\label{fig:Example1}
\end{figure}

By this definition we have satisfied all block constraints.
%%%\begin{equation} \label{eq:sol1}
%%%(\xi^*_{i,j},\xi^*_{i+1,j},\xi^*_{i,j+1},\xi^*_{i+1,j+1}):=
%%%\left\{\begin{array}{lll}
%%%(1,0,0,0) &:& (\sigma_i(j)\leq c_i) \wedge (\rho_j(i)\leq r_j),\\ 
%%%(0,1,0,0) &:& (\sigma_i(j)>    c_i) \wedge (\rho_j(i)\leq r_j),\\
%%%(0,0,1,0) &:& (\sigma_i(j)\leq c_i) \wedge (\rho_j(i)>    r_j),\\
%%%(0,0,0,1) &:& (\sigma_i(j)>    c_i) \wedge (\rho_j(i)>    r_j),
%%%\end{array}\right. 
%%%\end{equation} for every $(i,j)\in I,$ we satisfy all block constraints. 
A simple counting argument shows that the row and columns sums are also as required. In fact, for $j\in \Pi_y(I)$ and $i \in \Pi_x(I)$ we have
\[%\label{mono:eq1}
\begin{aligned}
&\sum_{\mathclap{p:(p,j)\in G(I)}}\xi^*_{p,j}&=&  &\left|\left\{p:\left((p,j)\in I\right)\wedge\left(\rho_j(p)\leq r_j\right)\right\} \right|& &=& &&r_{j},& &&&\\
&\sum_{\mathclap{p:(p,j)\in G(I)}}\xi^*_{p,j+1}&=& &\left|\left\{p:\left((p,j)\in I\right)\wedge\left(\rho_j(p)> r_j\right)\right\} \right|& &=& &&\rho_j(m)-r_j&=& &&r_{j+1},\\
&\sum_{\mathclap{q:(i,q)\in G(I)}}\xi^*_{i,q}&=& &\left|\left\{q:\left((i,q)\in I\right)\wedge\left(\sigma_i(q)\leq c_i\right)\right\} \right|& &=& &&c_{i},& &&&\\
&\sum_{\mathclap{q:(i,q)\in G(I)}}\xi^*_{i+1,q}&=& &\left|\left\{q:\left((i,q)\in I\right)\wedge\left(\sigma_i(q)> c_i\right)\right\} \right|& &=& &&\sigma_i(n)-c_i&=& &&c_{i+1}.
\end{aligned}
\] 
%\begin{equation}\label{mono:eq1}
%\begin{aligned}
%\sum_{\mathclap{(p,j)\in G(I)}}\xi^*_{p,j}&=\left|\left\{p:\left((p,j)\in I\right)\wedge\left(\rho_j(p)\leq r_j\right)\right\} \right|&&=r_{j},&\\
%\sum_{\mathclap{(p,j)\in G(I)}}\xi^*_{p,j+l}&=\left|\left\{p:\left((p,j)\in I\right)\wedge\left(\sum_{h=0}^{l-1}r_{j+h}<\rho_j(p)\leq \sum_{h=0}^{l}r_{j+h}\right)\right\} \right|&&=r_{j+l},&l\in[k-1],\\
%\sum_{\mathclap{(i,q)\in G(I)}}\xi^*_{i,q}&=\left|\left\{q:\left((i,q)\in I\right)\wedge\left(\sigma_i(q)\leq c_i\right)\right\} \right|&&=c_{i},&\\
%\sum_{\mathclap{(i,q)\in G(I)}}\xi^*_{i+l,q}&=\left|\left\{q:\left((i,q)\in I\right)\wedge\left(\sum_{h=0}^{l-1}c_{i+h}<\sigma_i(q)<\sum_{h=0}^lc_{i+h}\right)\right\} \right|&&=c_{i+l},& l\in[k-1].
%\end{aligned}
%\end{equation}
This concludes the proof of~(\ref{mono2}).

To show~(\ref{mono3}), we consider a solution $x^*$ to $\textsc{DR}(1).$ Suppose that~$r_j>0$ and $r_{j+1}>0$ for some $j\in \Pi_y(I).$ Then, there exist $i,l\in[m]$ with $\xi^*_{i,j}=\xi^*_{l,j+1}=1.$ By the block constraints we have $\xi^*_{i,j+1}=\xi^*_{l,j}=0.$ Applying a switch to $\xi^*_{i,j},$ $\xi^*_{l,j+1},$ $\xi^*_{i,j+1},$ and $\xi^*_{l,j},$ i.e., setting $\xi'_{i,j}:=\xi'_{l,j+1}:=0,$ $\xi'_{i,j+1}:=\xi'_{l,j}:=1,$ and $\xi'_{p,q}:=\xi^*_{p,q}$ otherwise, yields thus another solution. As the same argument holds also for the column sums we have shown that the condition in~(\ref{mono3}) is necessary. The condition is, on the other hand, also sufficient since the 0-values of the row and column sums leave only a single position $(p,q)$ in each block $B(i,j),$ $(i,j)\in I,$ for which $\xi_{p,q}$ can take on a non-zero value (in fact, the value needs to be~1 to satisfy the block constraints).

We now turn to~(\ref{mono1}). The algorithm checks the condition~(\ref{mono2}). If it is violated, it reports infeasibility; otherwise, a solution is constructed through~\eqref{eq:sol1}. Clearly, all steps can be performed in polynomial time.
\end{proof}

\begin{corollary} \label{cor:mono3}\hfill
\begin{enumerate}[(i)]
\item \label{cormono2} An instance $\mathcal{I}$ of \textsc{DR}$(3)$ is feasible if, and only if, for every $(i,j) \in I$ we have
\[r_j+r_{j+1}=3\rho_j(m) \quad \textnormal{ and } \quad c_i+c_{i+1}=3\sigma_i(n).\]
%\[\sum_{l=0}^{k-1}r_{j+l}=\rho_j(m) \quad \textnormal{ and } \quad \sum_{l=0}^{k-1}c_{i+l}=\sigma_i(n);\]
\item \label{cormono3} The solution of a feasible instance of \textsc{DR}$(3)$ is unique if, and only if, for every~$(i,j) \in I$ we have 
\[(2\rho_j(m)-r_j)\cdot (2\rho_j(m)-r_{j+1})=0 \qquad \textnormal{and}\qquad (2\sigma_i(n)-c_i)\cdot (2\sigma_i(n)-c_{i+1})=0.\]
\item \label{cormono1} \textsc{DR}$(3)\in\mathbb{P}.$ 
\end{enumerate}
\end{corollary}
\begin{proof}
The results follow directly from the results for $\textsc{DR}(1)$ in Lem.~\ref{lemma:mono} as the former are obtained from the latter by inversion and vice versa, i.e., by replacing $r_{j+l}$ by $2\rho_j(m)-r_{j+l}$ and $c_{i+l}$ by $2\sigma_i(n)-c_{i+l},$ $(i,j)\in I,$ $l\in\{0,1\},$ and, in the solutions, all $1$'s by $0$'s and $0$'s by $1$'s.
\end{proof}

Next we prove polynomial-time solvability of $\textsc{DR}(2).$ To this end we need two lemmas. The first can be viewed as dealing with a certain ``two-color version'' of discrete tomography. The second  deals with specific switches or interchanges (in the sense of \cite{ryser57}). 

\begin{lemma} \label{lem:unimod2}
The problem, given $m,n\in2\mathbb{N},$ $I\subseteq C(m,n,2),$ and $r'_j,c'_i\in\mathbb{N}_0,$ $(i,j)\in I,$ decide whether there exist non-negative integers $\zeta_{i,j},$ $\eta_{i,j},$ $(i,j) \in I,$ such that
\begin{alignat*}{4}
&\sum_{p:(p,j)\in I} \zeta_{p,j}&&=r'_j, &\hspace*{5ex}& j \in \Pi_y(I),\\
&\sum_{q:(i,q)\in I} \eta_{i,q}&&=c'_i, &&i\in \Pi_x(I),\\
&\zeta_{i,j}+\eta_{i,j}&&\leq 1, &&(i,j) \in I,
\end{alignat*} and, if so, determine a solution, can be solved in polynomial time.
\end{lemma}

\begin{proof}
We show that the problem can be formulated as a linear program involving a totally unimodular coefficient matrix (for background material, see~\cite{schrijver}). 

We assemble the variables $\zeta_{i,j},$ $\eta_{i,j}$ for $(i,j)\in  I$ into a $2\cdot|I|$-dimensional vector and rephrase the constraints in matrix form  $Ax= r',$ $A'x= c',$ and $(E\: E)x\leq \1,$ where~$r'$ contains the $r'_j$'s,  $c'$ contains the $c'_i$'s, and $E$ denotes the $|I|\times|I|$-identity matrix. The problem is then to determine an integer solution of $\{x\::\: Mx \leq b, \: x\geq 0\},$ where 
\[
M:=\left(\begin{array}{cc}A& 0 \\ -A&0\\0 & A' \\ 0 &-A'\\ E & E\end{array}\right) \quad \textnormal{and}\quad   b:=\left(\begin{array}{c}r'\\-r'\\\ c'\\-c'\\\1 \end{array}\right).
\]
The submatrix
\[
M':=\left(\begin{array}{cc}A& 0 \\ 0 & A' \\ E & E\end{array}\right) 
\] is totally unimodular as it is the node-edge incidence matrix of a bipartite graph (one of the two parts of the partition  is the set of the last $|I|$ rows of~$M'$). But then~$M$ is also totally unimodular as it results from appending a subset of rows of~$-M'$ to the totally unimodular matrix~$M'.$ In other words, $\{x\::\: Mx \leq b,\: x\geq 0\}$ is an integral polyhedron, i.e., it coincides with the convex hull of its integral vectors. Hence, a vertex of this polyhedron can be found in polynomial time by linear programming. 
\end{proof}

Next we study the structure of solutions in terms of the patterns of the blocks. This will allow us to focus on special classes of solutions, which we will later call \emph{reduced} solutions. The problem of finding reduced solutions to a given instance of $\textsc{DR}$ will decompose into the five independent subtasks of solving instances of $\textsc{DR}(\nu),$ $\nu\in\{0,\dots,4\}.$ The concept of reduced solutions will also play a role in our proof of Lem.~\ref{lem:dr2} showing that $\textsc{DR}(2)\in\mathbb{P}.$ 

Blocks in solutions are filled with zeros and ones. There are~16 such combinations. By grouping the combinations that have the same row sums, we distinguish the nine block types $\mathfrak{A}_1,$ $\mathfrak{A}_2,$ $\mathfrak{B}_1,$ $\mathfrak{B}_2,$ $\mathfrak{B}_3,$ $\mathfrak{C}_1,$ $\mathfrak{C}_2,$ $\mathfrak{D}_1,$ $\mathfrak{D}_2$ shown in Fig.~\ref{fig:types:1}. Some of the types consist of several subtypes/combinations: type~$\mathfrak{A}_i$ and $\mathfrak{C}_i$ consist of the two subtypes $\mathfrak{A}_{i,i'}$ and $,\mathfrak{C}_{i,i'}$ $i,i'\in[2],$ respectively, while $\mathfrak{B}_3$ consists of the subtypes $\mathfrak{B}_{3,1},\dots,\mathfrak{B}_{3,4}$. For any given solution, we will later count the number of occurrences of the different block types in each horizontal and vertical strip. The corresponding variables for the horizontal strip $[m]\times \{j,j+1\}$ are indicated in Fig.~\ref{fig:types:1}. For instance, $\alpha_j$ counts the number of occurrences of block type~$\mathfrak{A}_1,$ while $\alpha_{j+1}$ counts the number of occurrences of block type~$\mathfrak{A}_2.$

\begin{figure}[htb]
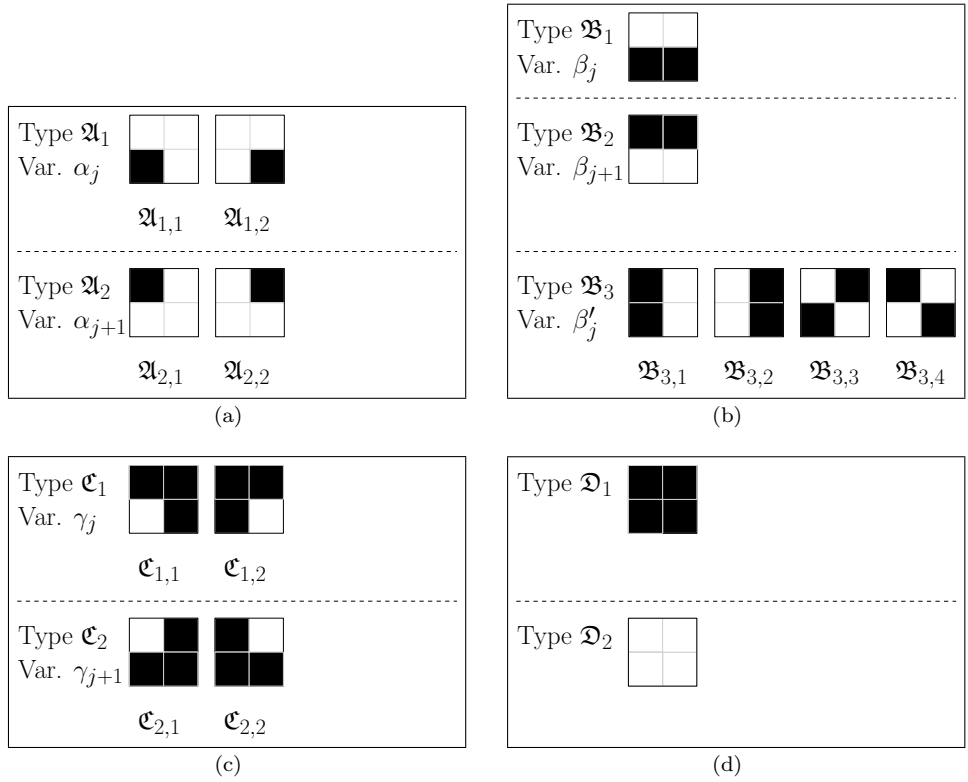
 
\centering
\subfigure[]{ \fbox{\includegraphics[width=0.37\textwidth]{Types1.pdf}}}\hspace*{2ex}
\subfigure[]{ \fbox{\includegraphics[width=0.37\textwidth]{Types3.pdf}}}\\
\subfigure[]{ \fbox{\includegraphics[width=0.37\textwidth]{Types4.pdf}}}\hspace*{2ex}
\subfigure[]{ \fbox{\includegraphics[width=0.37\textwidth]{Types5.pdf}}}
\caption{Block types $\mathfrak{A}_1,$ $\mathfrak{A}_2,$ $\mathfrak{B}_1,$ $\mathfrak{B}_2,$ $\mathfrak{B}_3,$ $\mathfrak{C}_1,$ $\mathfrak{C}_2,$ $\mathfrak{D}_1,$ $\mathfrak{D}_2.$  The indicated variables $\alpha_j,$ $\alpha_{j+1},$ $\beta_j,$ $\beta_{j+1},$ $\beta'_j,$ $\gamma_j,$ $\gamma_{j+1}$ count the number of occurrences of the respective block type in the  horizontal strip $[m]\times \{j,j+1\}.$ The block types in~(d) are not further considered as they can be eliminated in a preprocessing step. For referencing purposes, the figure also introduces the subtypes $\mathfrak{A}_{i_1,i_2}$ of $\mathfrak{A}_{i_1},$ $\mathfrak{B}_{3,i_3}$ of $\mathfrak{B}_{3},$ and $\mathfrak{C}_{i_1,i_2}$ of $\mathfrak{C}_{i_1},$ $i_1,i_2\in[2],$ $i_3\in[3].$}\label{fig:types:1}
\end{figure}

We introduce now the concept of local switches. It is based on the following observation. If, for instance, a horizontal strip contains a block of type $\mathfrak{A}_{1,i},$ $i\in[2],$ and another one of type $\mathfrak{B}_2$ we may replace them by a block of type $\mathfrak{A}_{2,i}$ and a block of type $\mathfrak{B}_{3,3},$ respectively, without changing the row and column constraints. Such a change will be called a \emph{horizontal local switch of class~$\langle1\rangle.$ } %We call such a switch a \emph{horizontal local switch of class~$\langle1\rangle.$} 
In a similar way we define \emph{horizontal local switches of class $\langle t\rangle,$ $t=2,\dots,7,$} according to Table~\ref{table:switchesexample}. Horizontal local switches of class~$\langle7\rangle$ affect only single blocks (blocks of type~$\mathfrak{B}_{3,4}$ are turned into blocks of type~$\mathfrak{B}_{3,3}$).

The (similar) switches in the vertical strips, called vertical local switches, are listed in Table~\ref{table:verticalswitches}.

\begin{table}[htb]
\begin{center}
{\small
\begin{tabular}{cp{9ex}p{1ex}p{7ex}lc}\toprule
Class &  Turns && Into &&Illustration\\\midrule
%& &&  && (horizontal local switch)& (vertical local switch) \\\midrule
$\langle1\rangle$ &$(\mathfrak{A}_{1,i}, \mathfrak{B}_2)$&$\rightarrow$&$(\mathfrak{A}_{2,i},\mathfrak{B}_{3,3})$ &&\begin{minipage}{.3\textwidth}\includegraphics[width=\textwidth]{switch2.pdf}\end{minipage}\\
&&&&&\\
$\langle2\rangle$ &$(\mathfrak{A}_{2,i},\mathfrak{B}_1)$&$\rightarrow$&$(\mathfrak{A}_{1,i},\mathfrak{B}_{3,3})$ &&\begin{minipage}{.3\textwidth}\includegraphics[width=\textwidth]{switch4.pdf}\end{minipage}\\
&&&&&\\
$\langle3\rangle$ &$(\mathfrak{B}_1, \mathfrak{B}_2)$&$\rightarrow$&$(\mathfrak{B}_{3,3},\mathfrak{B}_{3,3})$ &&\begin{minipage}{.3\textwidth}\includegraphics[width=\textwidth]{switch1.pdf}\end{minipage}\\
&&&&&\\
$\langle4\rangle$ &$(\mathfrak{C}_{2,i},\mathfrak{B}_2)$&$\rightarrow$&$(\mathfrak{C}_{1,i},\mathfrak{B}_{3,3})$ &&\begin{minipage}{.3\textwidth}\includegraphics[width=\textwidth]{switch3.pdf}\end{minipage}\\
&&&&&\\
$\langle5\rangle$ &$(\mathfrak{C}_{1,i},\mathfrak{B}_1)$&$\rightarrow$&$(\mathfrak{C}_{2,i},\mathfrak{B}_{3,3})$ &&\begin{minipage}{.3\textwidth}\includegraphics[width=\textwidth]{switch5.pdf}\end{minipage}\\
&&&&&\\
$\langle6\rangle$ &$(\mathfrak{A}_{1,i},\mathfrak{C}_{1,i'})$&$\rightarrow$&$(\mathfrak{A}_{2,i},\mathfrak{C}_{2,i'})$ &&\begin{minipage}{.3\textwidth}\includegraphics[width=\textwidth]{switch6.pdf}\end{minipage}\\
&&&&&\\
$\langle7\rangle$&$(\mathfrak{B}_{3,4})$&$\rightarrow$&$(\mathfrak{B}_{3,3})$ &&\begin{minipage}{.3\textwidth}\includegraphics[width=\textwidth]{switch8.pdf}\end{minipage}\\
\bottomrule\\[-.1ex]
\end{tabular}}
\end{center}
\caption{Horizontal local switches. The parameters $i$ and $i'$ are elements of $[2].$}\label{table:switchesexample}
\end{table}

\begin{table}[htb]
\begin{center}
{\small
\begin{tabular}{clp{1ex}lc}\toprule
Class & Turns && Into& Illustration\\\midrule
$\langle1\rangle$ & $(\mathfrak{B}_{3,2},\mathfrak{A}_{i,1})$ &$\rightarrow$& $(\mathfrak{B}_{3,3},\mathfrak{A}_{i,2})$ & \begin{minipage}{.15\textwidth}\includegraphics[width=\textwidth]{switch2v.pdf}\end{minipage}\\
&&&&\\
$\langle2\rangle$ & $(\mathfrak{B}_{3,1},\mathfrak{A}_{i,2})$ &$\rightarrow$& $(\mathfrak{B}_{3,3},\mathfrak{A}_{i,1})$ & \begin{minipage}{.15\textwidth}\includegraphics[width=\textwidth]{switch4v.pdf}\end{minipage}\\
&&&&\\
$\langle3\rangle$ & $(\mathfrak{B}_{3,2},\mathfrak{B}_{3,1})$ &$\rightarrow$& $(\mathfrak{B}_{3,3},\mathfrak{B}_{3,3})$ & \begin{minipage}{.15\textwidth}\includegraphics[width=\textwidth]{switch1v.pdf}\end{minipage}\\
&&&&\\
$\langle4\rangle$ & $(\mathfrak{B}_{3,2},\mathfrak{C}_{i,2})$ &$\rightarrow$& $(\mathfrak{B}_{3,3},\mathfrak{C}_{i,1})$ & \begin{minipage}{.15\textwidth}\includegraphics[width=\textwidth]{switch3v.pdf}\end{minipage}\\
&&&&\\
$\langle5\rangle$ & $(\mathfrak{B}_{3,1},\mathfrak{C}_{i,1})$ &$\rightarrow$& $(\mathfrak{B}_{3,3},\mathfrak{C}_{i,2})$ & \begin{minipage}{.15\textwidth}\includegraphics[width=\textwidth]{switch5v.pdf}\end{minipage}\\
&&&&\\
$\langle6\rangle$ & $(\mathfrak{C}_{i',1},\mathfrak{A}_{i,1})$ &$\rightarrow$& $(\mathfrak{C}_{i',2},\mathfrak{A}_{i,2})$ & \begin{minipage}{.15\textwidth}\includegraphics[width=\textwidth]{switch6v.pdf}\end{minipage}\\
&&&&\\
$\langle7\rangle$ & $(\mathfrak{B}_{3,4})$ &$\rightarrow$& $(\mathfrak{B}_{3,3})$ & \begin{minipage}{.15\textwidth}\includegraphics[width=\textwidth]{switch8v.pdf}\end{minipage}\\\bottomrule\\[-.1ex]
\end{tabular}}
\end{center}
\caption{Vertical local switches. The parameters $i$ and $i'$ are elements of $[2].$}\label{table:verticalswitches}
\end{table}

A switch is called \emph{local} if it is a horizontal or vertical local switch of some class~$\langle t\rangle,$ $t\in[7].$ A solution is called \emph{reduced} if no local switch of any class~$\langle t \rangle,$ $t\in[7],$ can be applied. Note that local switches are directed in the sense that they turn the types of  blocks listed in the second column of Tables~\ref{table:switchesexample} and~\ref{table:verticalswitches} into the types given in the respective third column of the tables. Oppositely directed switches, called \emph{reversed local switches}, will be used later when we discuss questions of uniqueness.

%\newpage

\begin{lemma} \label{lem:reduction}
\hfill
\begin{enumerate}[(i)]
\item \label{red1} Application of a local switch to a solution of a given instance of \textsc{DR} yields again a solution of the same instance.
\item \label{red2} An instance of $\textsc{DR}$ has a solution if, and only if, there is a reduced solution.
\end{enumerate}
\end{lemma}
\begin{proof}
To prove~(i) just observe that the local switches do neither change the row and column sums nor the number of ones contained in each block.

Let us now turn to~(ii). Suppose the given instance has a solution~$x^*.$ The local switches of class~$\langle t\rangle,$ $t\in[7]\setminus\{6\},$ increase the number of blocks $B(i,j)$ of type $\mathfrak{B}_{3,3},$ i.e., \[\xi^*_{i,j}=\xi^*_{i+1,j+1}=1 \qquad \textnormal{and}\qquad  \xi^*_{i+1,j}=\xi^*_{i,j+1}=0.\] The process of applying local switches of class~$\langle t \rangle,$ $t\in[7]\setminus\{6\},$ thus needs to terminate since the number of blocks in each problem instance is finite. Further, note that local switches of class~$\langle6\rangle$ applied in horizontal strips increase the number of blocks of type~$\mathfrak{A}_2.$   The vertical switches  of that class do not decrease the number of such blocks, but increase the number of blocks which are of type $\mathfrak{A}_{1,2}$ or $\mathfrak{A}_{2,2}.$ Hence, after finitely many switches, a given solution is converted into a reduced one.
\end{proof}

\begin{lemma}\label{lem:dr2}
$\textsc{DR}(2)\in\mathbb{P}.$
\end{lemma}
\begin{proof}
Our general strategy is as follows. First we show that a given instance is feasible if, and only if, there exists a specific reduced solution that contains only three block types. Then we give a polynomial-time algorithm that finds such a solution or determines infeasibility.

%In the following we will again use the notational convention that specific settings of variables and parameters are signified by a superscript $^*.$

By rearranging the rows and columns, if necessary, we may assume that
\[ %\label{type:eq0}
r_{j}\geq r_{j+1} \qquad \textnormal{and}\qquad c_{i}\geq c_{i+1}, \quad (i,j)\in I.
\]

Now, first suppose that the given instance has a solution. Then, by Lem.~\ref{lem:reduction} there exists a reduced one. In the following, we consider such a reduced solution. %We further assume that this solution contains no block with $\xi^*_{i,j}=\xi^*_{i+1,j+1}=0$ and $\xi^*_{i,j+1}=\xi^*_{i+1,j}=1$ (which can be ensured by  applying a switch to the points of a single block).

The potential block types, and corresponding variables $\beta_j,$ $\beta_{j+1},$ and $\beta'_j$ counting the number of occurrences of the respective block types in the horizontal strip $[m]\times\{j,j+1\}$ of a reduced solution, are shown in Fig.~\ref{fig:types:1}(b). Block types counted by the same variable have equal row sums.

Consider now for fixed $(i,j)\in I$ the blocks in $G(I)\cap ([m]\times\{j,j+1\}).$  Counting the block types, any solution needs to satisfy
\[
\begin{array}{cccccc}
2\beta_j  &            &+& \beta'_j &=& r_j,\\
          &2\beta_{j+1}&+& \beta'_j &=& r_{j+1},
\end{array}
\]
which implies 
 %\label{eq:beta}
\begin{align}
\beta_{j}-\beta_{j+1}&=(r_j-r_{j+1})/2. \label{eq:beta1s}
\end{align}
As the solution is reduced no horizontal local switch of class~$\langle3\rangle$ can be applied, hence it is impossible that both $\beta_j>0$ and $\beta_{j+1}>0.$ %we need to have $\beta_j=0$ or $\beta_{j+1}=0.$ 
This, together with $r_j\geq r_{j+1}$ and~\eqref{eq:beta1s}, implies \[\beta_{j+1}=0, \qquad \beta_j=(r_j-r_{j+1})/2,\quad \textnormal{and} \quad \beta'_j=r_{j+1}.\] 
(Note that by $\beta_{j+1}=0$ there is no block of type~$\mathfrak{B}_2.$ As no horizontal local switch of class~$\langle7\rangle$ can be applied, there is neither a block of type~$\mathfrak{B}_{3,4}$.)

The same argument applies to the other horizontal strips, and a similar argument holds for the vertical strips (which, in particular, rules out blocks of type $\mathfrak{B}_{3,2}$). Hence, we conclude that, due to reduction, $\textsc{DR}(2)$ amounts to the task of assigning  one of the three block types $\mathfrak{B}_1,\mathfrak{B}_{3,1},\mathfrak{B}_{3,3}$ shown in Fig.~\ref{fig:types:2} to each of the blocks $B(i,j),$ $(i,j)\in I,$ such that, in particular, the number of type~$\mathfrak{B}_1$ and type~$\mathfrak{B}_{3,1}$ blocks in each horizontal and vertical strip, respectively, equals some prescribed cardinality.

\begin{figure}[htb] 
\centering
\includegraphics[width=0.35\textwidth]{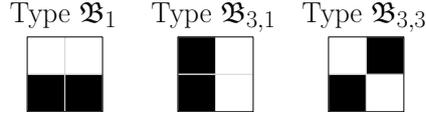}
\caption{The three potential block types in solutions of $\textsc{DR}(2)$ instances considered in the proof of Lem.~\ref{lem:dr2}. Value~$1$ and $0$ are indicated in black and white, respectively.}\label{fig:types:2}
\end{figure}

Now, we introduce $0/1$-variables $\zeta_{i,j},$ $\eta_{i,j},$ $\lambda_{i,j},$ $(i,j)\in I,$ that are one if, and only if,~$B(i,j)$ is of type~$\mathfrak{B}_1,$ $\mathfrak{B}_{3,1},$ $\mathfrak{B}_{3,3},$  respectively. Then we ask for $0/1$-solutions that satisfy

\[
\begin{array}{lllll}
\sum\limits_{p:(p,j)\in I} \zeta_{p,j}&=&(r_j-r_{j+1})/2,&&j \in \Pi_y(I),\\[2.8ex]
\sum\limits_{p:(p,j)\in I} (\eta_{p,j}+\lambda_{p,j})&=&\rho_j(m)-(r_j-r_{j+1})/2,&&j \in \Pi_y(I),\\[2.8ex]
\sum\limits_{q:(i,q)\in I} \eta_{i,q}&=&(c_i-c_{i+1})/2, &&i \in \Pi_x(I),\\[2.8ex]
\sum\limits_{q:(i,q)\in I} (\zeta_{i,q}+\lambda_{i,q})&=&\sigma_i(n)-(c_i-c_{i+1})/2, &&i \in \Pi_x(I),\\[2.8ex]
\zeta_{i,j}+\eta_{i,j}+\lambda_{i,j}&=& 1, &&(i,j) \in I.
\end{array} 
\] Eliminating $\lambda_{i,j}=1-\zeta_{i,j}-\eta_{i,j},$ $(i,j)\in I,$ we see that the $0/1$-solutions of this system are in one-to-one correspondence with the $0/1$ solution of
\begin{equation}
\begin{array}{lllll}
\sum\limits_{p:(p,j)\in I} \zeta_{p,j}&=&(r_j-r_{j+1})/2,&&j \in \Pi_y(I),\\[2.8ex]
\sum\limits_{q:(i,q)\in I} \eta_{i,q}&=&(c_i-c_{i+1})/2, &&i \in \Pi_x(I),\\[2.8ex]
\zeta_{i,j}+\eta_{i,j}&\leq& 1, &&(i,j) \in I.
\end{array} \label{eq:zetasol}
\end{equation}
As we assume that our $\textsc{DR}(2)$ instance is feasible (hence containing two ones in each block), we have $r_j+r_{j+1}=2\rho_j(m)$ and $c_i+c_{i+1}=2\sigma_i(n)$ for every $j \in \Pi_y(I)$ and $i \in \Pi_x(I).$ Hence, if a solution to~\eqref{eq:zetasol} is expanded, replacing the $\zeta^*_{i,j},$ $\eta^*_{i,j}, 1-\zeta^*_{i,j}-\eta^*_{i,j},$ $(i,j)\in I,$  by the respective block types, then, indeed 
\[
\begin{array}{lllll}
\sum\limits_{p:(p,j)\in G(I)}\xi^*_{p,j}&=&2\cdot(r_j-r_{j+1})/2+\rho_j(m)-(r_j-r_{j+1})/2&=&r_j,\\[2.8ex]
\sum\limits_{p:(p,j)\in G(I)}\xi^*_{p,j+1}&=&\rho_j(m)-(r_j-r_{j+1})/2&=&r_{j+1},\\[2.8ex]
\sum\limits_{q:(i,q)\in G(I)}\xi^*_{i,q}&=&2\cdot(c_i-c_{i+1})/2+\sigma_i(n)-(c_i-c_{i+1})/2&=&c_i,\\[2.8ex]
\sum\limits_{q:(i,q)\in G(I)}\xi^*_{i+1,q}&=&\sigma_i(n)-(c_i-c_{i+1})/2&=&c_{i+1},\\[2.8ex]
\end{array}
\] for every $j \in \Pi_y(I)$ and $i \in \Pi_x(I).$ Any $0/1$-solution to \eqref{eq:zetasol} provides thus a solution to the given $\textsc{DR}(2)$ instance. Since $0/1$-solutions to \eqref{eq:zetasol} can be found in polynomial time (see Lem.~\ref{lem:unimod2}), we have hence determined a solution to our instance in polynomial time. Note that this solution is reduced as it does not contain blocks of type~$\mathfrak{B}_{2},$ $\mathfrak{B}_{3,2},$ and $\mathfrak{B}_{3,4}.$

It remains to consider the case that the given \textsc{DR}$(2)$ instance is infeasible. Then one of the previous steps must fail: either~\eqref{eq:zetasol} admits no $0/1$-solution or the expanded ``solution'' does not satisfy all row and column sums. Both can be detected in polynomial time (see also Lem.~\ref{lem:unimod2}).
\end{proof}

Let us now turn to the task of solving \textsc{DR}. Suppose, we are given a feasible instance of \textsc{DR} that does not contain any block constraints of the form $v(i,j)=\nu,$ $\nu\in\{0,4\},$ and whose row and column sums satisfy \begin{equation} \label{type:eq0}
r_{j}\geq r_{j+1} \qquad \textnormal{and}\qquad c_{i}\geq c_{i+1}, \quad (i,j)\in C(m,n,2).
\end{equation} We call such an instance a \emph{proper} instance.

In the following, we consider for fixed $j\in[n]\cap(2\mathbb{N}_0+1)$ the horizontal strip $[m]\times \{j,j+1\}$. We set 
\begin{align*}
\vja:=|\{i\in [m]\cap(2\mathbb{N}_0+1):v(i,j)=1\}|,\\
\vjb:=|\{i\in [m]\cap(2\mathbb{N}_0+1):v(i,j)=2\}|,\\
\vjc:=|\{i\in [m]\cap(2\mathbb{N}_0+1):v(i,j)=3\}|.
\end{align*}  

Let us now consider a solution of the proper instance. A count of the block types of the solution in that strip (the variables for the block types are indicated in Fig.~\ref{fig:types:1}), yields
\begin{alignat}{15}
\alpha_j&& &&             &&\:\:+\:\:&&2\beta_j&& &&             &&\:\:+\:\:&&\beta'_j  &&\:\:+\:\:&& \gamma_j  &&\:\:+\:\:&&2\gamma_{j+1} &&\:\:=\:\:&&&r_j,\nonumber \\ 
        && &&\alpha_{j+1} &&\:\:+\:\:&&        && &&2\beta_{j+1} &&\:\:+\:\:&&\beta'_j  &&\:\:+\:\:&& 2\gamma_j &&\:\:+\:\:&&\gamma_{j+1}  &&\:\:=\:\:&&&r_{j+1}, \nonumber \\ 
\alpha_j&&\:\:\:+\:\:\:&&\alpha_{j+1} && &&        && &&             && &&                 && &&           && &&              &&\:\:=\:\:&&&\vja, \label{eq:together1}\\
        && &&             && &&\beta_j &&\:\:+\:\:&& \beta_{j+1} &&\:\:+\:\:&& \beta'_j && &&           && &&              &&\:\:=\:\:&&&\vjb,\label{eq:together2}\\ 
        && &&             && &&        && &&             && &&                 && &&  \gamma_j &&\:\:+\:\:&&\gamma_{j+1}  &&\:\:=\:\:&&&\vjc.\label{eq:together3}
\end{alignat}

From the system we obtain
\begin{alignat}{11}
 \alpha_j &&\:\:+\:\:&& \beta_{j} &&\:\:-\:\:&& \beta_{j+1}&&         &&        &&\:\:+\:\:&& \gamma_{j+1} &&\:\:\:=\:\:\:&&&r_j-\vjb-\vjc,\label{type:eq:main1}\\
-\alpha_j &&\:\:-\:\:&& \beta_{j} &&\:\:+\:\:&& \beta_{j+1}&&\:\:+\:\:&&\gamma_j&& &&                      &&\:\:\:=\:\:\:&&&r_{j+1}-\vja-\vjb-\vjc.\label{type:eq:main3}
\end{alignat}

We will show that for reduced solutions, the variable values of $\alpha_j,$ $\alpha_{j+1},$ $\beta_{j},$ $\beta'_j,$ $\beta_{j+1},$ $\gamma_j,$ $\gamma_{j+1}$ are uniquely determined and can, actually, be computed from the data (i.e., from $r_{j},$ $r_{j+1},$ $\vja,$ $\vjb,$ and $\vjc$). This will then allow us to split the task into instances of $\textsc{DR}(\nu),$ $\nu\in[3].$

We begin with a simple observation.

\begin{lemma}
Any reduced solution to a proper instance of \textsc{DR} satisfies $\beta_{j+1}=0.$
\end{lemma}
\begin{proof} Suppose there is a reduced solution with $\beta_{j+1}>0.$ As we cannot apply any further horizontal local switch of class~$\langle t \rangle$, $t\in \{1,3,4\},$ we need to have $\alpha_j=\beta_j=\gamma_{j+1}=0,$ hence by~\eqref{eq:together1} and~\eqref{eq:together3} $\alpha_{j+1}=\vja,$ and $\gamma_j=\vjc.$ Using~\eqref{type:eq:main1} and~\eqref{type:eq:main3} we obtain 
\begin{alignat*}{3}
&r_j&&\:\:\:=\:\:\:&&\vjb+\vjc-\beta_{j+1},\\
&r_{j+1}&&\:\:\:=\:\:\:&&\vja+\vjb+2\vjc+\beta_{j+1},
\end{alignat*} which imply $r_j<\vjb+\vjc\leq r_{j+1},$ a contradiction to~\eqref{type:eq0}.
\end{proof}

Given a proper instance, we distinguish three cases.

\begin{itemize}
\item Case~1: Every reduced solution satisfies $\beta_{j}>0.$\\[-2ex]
\item Case~2: There is a reduced solution satisfying $\beta_j=0,$ and in any such solution we have $\alpha_j>0.$\\[-2ex]
\item Case~3: There is a reduced solution satisfying $\beta_j=\alpha_j=0.$
\end{itemize}

\begin{lemma} \label{lem:cases}
Given a proper instance of~\textsc{DR}. Then,
\begin{alignat*}{8}
 \textnormal{Case~1}\qquad &&\Leftrightarrow&&\qquad  &\vjc &&\leq&&\:\: r_{j+1}\:\: &&<&&\:\:\vjb+\vjc,\\
 \textnormal{Case~2}\qquad &&\Leftrightarrow&&\qquad  &\vjb+\vjc &&\leq&&\:\: r_{j+1}\:\: &&<&&\:\:\vja+\vjb+\vjc,\\
 \textnormal{Case~3}\qquad &&\Leftrightarrow&&\qquad  &\vja+\vjb+\vjc &&\leq&& \:\:r_{j+1}\:\: &&\leq&&\:\: \vja+\vjb+2\vjc.
\end{alignat*}
and in all cases the variable values of $\alpha_j,$ $\alpha_{j+1},$ $\beta_{j},$ $\beta'_j,$ $\beta_{j+1},$ $\gamma_j,$ $\gamma_{j+1}$ are uniquely determined by the data.
\end{lemma}
\begin{proof}By Lem.~\ref{lem:reduction}, a proper instance has a reduced solution. Hence, one of the disjoint Cases~$1-3$ occurs. Once we have established the implications ``$\Rightarrow$''  the reverse implications  follow, because the stated intervals for the~$r_j$ are disjoint (thereby uniquely specifying the case).
 
Hence it suffices to prove the implications ``$\Rightarrow$.'' Let us begin with Case~1. As we cannot apply any further horizontal local switch of class~$\langle t\rangle,$ $t\in\{2,5\},$  we  must  have $\alpha_{j+1}=\gamma_j=0,$  hence $\alpha_j=\vja,$ and $\gamma_{j+1}=\vjc.$ 
Eq.~\eqref{type:eq:main3} implies $\beta_j=\vjb+\vjc-r_{j+1}.$ As $0<\beta_j\leq \vjb,$ we thus have $\vjc\leq r_{j+1}<\vjb+\vjc.$ Further,~\eqref{eq:together2} yields $\beta'_j=r_{j+1}-\vjc.$ Hence, all values of the variables are uniquely determined.

In Case~2, as $\beta_{j+1}=\beta_j=0,$ we have $\beta'_j=\vjb,$ and as we cannot apply any further horizontal local switch of class~$\langle 6\rangle,$ we need to have $\gamma_j=0,$  hence $\gamma_{j+1}=\vjc.$ Eq.~\eqref{type:eq:main3} implies $\alpha_j=\vja+\vjb+\vjc-r_{j+1}.$ As $0<\alpha_j\leq \vja,$ we thus have $\vjb+\vjc\leq r_{j+1}<\vja+\vjb+\vjc.$ Further,~\eqref{eq:together1} yields $\alpha_{j+1}=r_{j+1}-\vjb-\vjc.$ Hence, again, all variables are determined.

In Case~3, as $\beta_{j+1}=\beta_j=\alpha_j=0,$ we have $\beta'_j=\vjb$ and $\alpha_{j+1}=\vja.$ Eq.~\eqref{type:eq:main3} implies $\gamma_j=r_{j+1}-\vja-\vjb-\vjc.$  As $0\leq\gamma_j\leq \vjc$ we thus have $\vja+\vjb+\vjc\leq r_{j+1}\leq \vja+\vjb+2\vjc.$ Further,~\eqref{eq:together3} yields $\gamma_{j+1}=\vja+\vjb+2\vjc-r_{j+1},$ and all variable values are determined.
\end{proof}

We are now prepared to give a proof of Thm.~\ref{thm:main1}.

\begin{proof}[Proof of Thm.~\ref{thm:main1}]
We show that the task of finding a solution to a given problem instance can be decomposed into five independent subtasks of solving instances of $\textsc{DR}(\nu),$ $\nu\in\{0,\dots,4\}.$ The subtasks are then polynomial time solvable by Lem.~\ref{lemma:mono}, Cor.~\ref{cor:mono3} and~\ref{lem:dr2} (and trivially for $\textsc{DR}(0)$ and $\textsc{DR}(4)$). The key step in our proof is to deduce from the data the row and column sums for the particular subtasks. This deduction is possible for reduced solutions. 

We start with some preprocessing if necessary. First, if the problem instance contains block constraints of the form $v(i,j)=\nu,$ $\nu\in\{0,4\},$ we need to have $\xi_{i,j}=\xi_{i+1,j}=\xi_{i,j+1}=\xi_{i+1,j+1}=\nu/4.$ In this case, we fix the values of these variables, reduce the row and column sums accordingly, and consider only the remaining blocks. For notational convenience, we assume from now on that our instance does not contain such block constraints. 

Second, by possibly rearranging the rows and columns, we assume that~\eqref{type:eq0} holds.

Suppose, the problem instance is feasible, hence proper. (We return to the infeasible case later.) The instance has a solution and, hence by Lem.~\ref{lem:reduction}, a reduced one. Thus, considering a horizontal strip $[m]\times \{j,j+1\}$ with $j\in[n]\cap(2\mathbb{N}_0+1),$ one of the Cases~1,~2 or~3 needs to occur. Which of the cases occurs is, by Lem.~\ref{lem:cases}, determined by the data, simply by testing whether we have 
\begin{equation}
\begin{aligned}
&\vjc &&\leq&& r_{j+1} &&<&&\vjb+\vjc,&\quad \textnormal{or} \\
&\vjb+\vjc &&\leq&& r_{j+1} &&<&&\vja+\vjb+\vjc, &\quad \textnormal{or} \\ 
&\vja+\vjb+\vjc &&\leq&& r_{j+1} &&\leq&& \vja+\vjb+2\vjc.&
\end{aligned} \label{eq:intervals}
\end{equation}

In all three cases we  know the unique values of the $\alpha_j,$ $\alpha_{j+1},$ $\beta_j,$ $\beta_{j+1},$ $\beta'_j,$  $\gamma_j,$ and $\gamma_{j+1}.$ Therefore we know for each $\nu\in[3]$ the row sums 
\[\sum_{p:(p,j)\in G(I(\nu))}\xi_{p,j} \qquad \textnormal{and}\qquad \sum_{p:(p,j)\in G(I(\nu))}\xi_{p,j+1},\] with
 \[I(\nu):=\{(i,j)\in C(m,n,2):  v(i,j)=\nu\}\] for $G(I(\nu))$ according to~\eqref{eq:Gdef}. In particular, we have
\begin{equation} \label{eq:redxrays}
\begin{array}{lclllcl}
\sum\limits_{p:(p,j)\in G(I(1))}\xi_{p,j}&=&\alpha_j, &&\sum\limits_{p:(p,j)\in G(I(1))}\xi_{p,j+1}&=&\alpha_{j+1},\\[2.8ex]
\sum\limits_{p:(p,j)\in G(I(2))}\xi_{p,j}&=&2\beta_j+\beta'_j, &&\sum\limits_{p:(p,j)\in G(I(2))}\xi_{p,j+1}&=&\beta'_j,\\[2.8ex]
\sum\limits_{p:(p,j)\in G(I(3))}\xi_{p,j}&=&\gamma_j+2\gamma_{j+1}, && \sum\limits_{p:(p,j)\in G(I(3))}\xi_{p,j+1}&=&2\gamma_{j}+\gamma_{j+1}.
\end{array}
\end{equation}

Similarly, we obtain the individual vertical constraints. We can therefore decompose the reconstruction problem into the three independent subproblems $\textsc{DR}(\nu),$ $\nu\in[3],$  which, by Lem.~\ref{lemma:mono}, Cor.~\ref{cor:mono3}, and Lem.~\ref{lem:dr2}, are solvable in polynomial time. 

It remains to consider the case that the given instance of~$\textsc{DR}$ is infeasible. In this case, the decomposition into the subproblems needs to yield either an infeasible subproblem or the returned ``solution'' needs to violate one of the constraints. As both are detected in polynomial time, we have concluded the proof of this theorem.
\end{proof}

Note that the above proof is constructive. A corresponding polynomial-time algorithm for solving~\textsc{DR} is summarized in Algorithm~\ref{alg:alg1}.
 
\begin{algorithm}
\caption{Compute a solution to \textsc{DR}}
\label{alg:alg1}
\begin{algorithmic}[1]
\STATE{$I{(\nu)}:=\{(i,j)\in C(m,n,2):v(i,j)=\nu\}$, $\nu\in[4]_0.$}
\STATE{Solve \textsc{DR}(0) with $I:=I{(0)}$ by filling the $B_2(i,j),$ $(i,j)\in I,$ with zeros.}
\STATE{Solve \textsc{DR}(4) with $I:=I{(4)}$ by filling the $B_2(i,j),$ $(i,j)\in I,$ with ones.}
\STATE{Reduce the row and column sums accordingly.}
\STATE{Determine the row and column sums for the \textsc{DR}$(\nu)$ instances, $\nu\in[3]$ \\ (row sums according to~\eqref{eq:intervals}, \eqref{eq:redxrays}; column sums analogously).}
\STATE{Solve \textsc{DR}(1) with $I:=I{(1)}$ according to~\eqref{eq:sol1}.}
\STATE{Solve \textsc{DR}(2) with $I:=I{(2)}$ according to Lem.~\ref{lem:dr2}.}
\STATE{Solve \textsc{DR}(3) with $I:=I{(3)}$ according to Cor.~\ref{cor:mono3}.}
\IF{the returned $\xi^*_{p,q}, (p,q)\in [m]\times[n],$ satisfy all constraints}\STATE \textbf{return} $\xi^*_{p,q}, (p,q)\in [m]\times[n].$\ELSE \STATE \textbf{return} ``Instance is infeasible.''\ENDIF
\end{algorithmic}
\end{algorithm}

We also remark that the solutions returned by the proposed algorithm are always reduced. This is easily seen by~(i) verifying that the values $\alpha_j,$ $\alpha_{j+1},$ $\beta_j,$ $\beta_{j+1},$ $\beta'_j,$  $\gamma_j,$ and $\gamma_{j+1},$ $j\in[n],$  in each of the three cases in~\eqref{eq:intervals} (and correspondingly for the variables in the vertical strips) do not allow an application of any of the local switches of class~$\langle t\rangle,$ $t\in[6],$ because for each of the local switches there is always a corresponding variable of value zero; and~(ii) by noting that the proposed algorithm for $\textsc{DR}(2)$ returns, by definition, only solutions where no local switch of class~$\langle 7\rangle$ can be applied.

Now we turn to  uniqueness.
\begin{proof}[Proof of Thm.~\ref{main:unique}] 
By Thm.~\ref{thm:main1}, either infeasibility is detected or a reduced solution of an instance is determined in polynomial time. We can thus assume that the problem instance has a solution. 

Clearly, a solution is unique if, and only if (Condition~1:) there is only one reduced solution, and (Condition~2:) every solution is a reduced solution.

The solution returned by the algorithm is reduced. Any two reduced solutions define the same instances of the subproblems $\textsc{DR}(\nu),$ $\nu\in[3],$ because the problem instances are defined by means of~\eqref{eq:intervals}, i.e., the definition depends only on the particular values of the $r_1,\dots,r_n,$ and $c_1,\dots,c_m$ (note that $\vja,$ $\vjb,$ $\vjc$ can be directly determined from the input). Hence, there is only one reduced solution if, and only if, the solution to each of the subproblems $\textsc{DR}(\nu),$ $\nu\in[3],$ is unique.

The conditions in Lem.~\ref{lemma:mono}(\ref{mono3}) and Cor.~\ref{cor:mono3}(\ref{cormono3}), and hence uniqueness for the instances of $\textsc{DR}(1)$ and $\textsc{DR}(3),$ can be checked in polynomial time. For $\textsc{DR}(2)$ we have by the linear program considered in Lem.~\ref{lem:unimod2}, a basic feasible solution $v^*$ of the linear program, i.e., a vertex of the feasible region, at our disposal. Now we minimize the linear objective function $f(x)=x^Tv^*$ over the same feasible region, which again can be done in polynomial time. The solution $v^*$ is unique if, and only if, $f(x^*)={v^*}^Tv^*.$ (Note that multiple solutions contain the same number of ones as this is given by the sum of the row and column sums, respectively.) Hence uniqueness for~$\textsc{DR}(2)$, and therefore Condition~1, can be checked in polynomial time. 

Suppose now that the reduced solution~$x^*$ returned by the algorithm is unique among all reduced solutions, i.e.,~Condition~1 is satisfied. We have shown in Lem.~\ref{lem:reduction} that every solution can be reduced by applying a sequence of local switches. Reversing the sequence and the local switches, we thus see that Condition~2 holds if, and only if, there is no reversed local switch that can be applied to $x^*.$ There are $O(m^2n^2)$ possible pairs of blocks that need to be checked to form a reversed local switch, hence Condition~2 can also be checked in polynomial time.
\end{proof}

%-------------------------------------------------------------------------------------------------------------------------
\section{Data Uncertainty}\label{sect:4}
In the proof of Thm.~\ref{thm:main2} we use a transformation from the following $\mathbb{N}\mathbb{P}$-complete problem (see~\cite{cook1971}) \\

\begin{center}
\begin{minipage}{0.95\textwidth}
\textsc{1-In-3-SAT}\\[0.8ex]
\hspace*{1ex}\begin{tabular}{l@{  }l}
Instance: & \begin{minipage}[t]{0.82\textwidth} 
Positive integers $S,$ $T,$ and a set $\mathcal{C}$ of $S$ clauses over $T$ variables $\tau_1,\dots,\tau_T,$ where each clause consists of  three literals involving three different variables.
\end{minipage}\\ 
&\\[-1.5ex]
Task: &  \begin{minipage}[t]{0.82\textwidth} 
Decide whether there exists a satisfying truth assignment for $\mathcal{C}$ that sets exactly one literal true in each clause.  
\end{minipage}
\end{tabular}
\end{minipage}
\end{center}\vspace*{2ex}

For a given instance of \textsc{1-In-3-SAT} we will construct a circuit board that contains an initializer, several connectors, and clause chips. A truth assignment is transmitted through the circuit board, the clause chips ensure that the clauses are satisfied. 

The structure of our proof is in some ways similar to the proof from~\cite{ggp-99} that establishes the $\mathbb{N}\mathbb{P}$-hardness of the task of reconstructing lattice sets from X-rays taken in three or more directions.
 However, we need to deviate from \cite{ggp-99} in several key aspects, because the block constraints do not give a direct way of controlling points over large distances, and the elimination of ``unwanted solutions'' inside individual blocks seems also problematic. A fundamental difference to~\cite{ggp-99} is, for instance, that we encode the Boolean values for the variables by specific types of blocks, and the satisfiability of the clauses is verified via row and column sums. At first glance it seems that part of the variable assignment information is lost after verification, but it turns out that we can recover it from ``redundant'' information in our encoding.

\emph{Key ideas of the proof of Thm.~\ref{thm:main2}.}
Before we start with the detailed proof, we illustrate the general ideas by means of an example. Suppose, the instance $\mathcal{I}$ of \textsc{1-In-3-SAT} is given by $\mathcal{I}:=(S,T,\mathcal{C}):=(1,4,\{\tau_1\vee \lnot\tau_2 \vee \tau_3\}).$ 
We will define an instance $\mathcal{I}'$ of \textsc{nDR}$(\varepsilon)$ such that there is a solution for $\mathcal{I}'$ if, and only if, there is one for $\mathcal{I}.$ To this end, we construct a circuit board contained in the box $[34]\times[34].$ With each $(p,q)\in [34]\times[34]$ we associate a variable $\xi_{p,q}.$ Again, the figures will show the associated pixels $(p,q)+[0,1]^2;$ see Fig.~\ref{fig:gridimage}.

The circuit board consists of three major types of components: An \emph{initializer}, \emph{connectors}, and \emph{clause chips}. The initializer and the connectors are rather similar. The initializer contains for every variable $\tau_t,$ $t\in[T],$ a so-called \emph{$\tau_t$-chip}, the connectors contain for every variable a \emph{$\lnot\tau_t$-chip.} The clause chips are more complex as they consist of two \emph{collectors}, two \emph{verifiers}, and a \emph{transmitter}. 

Figure~\ref{fig:P2example}(a) shows the circuit board for our instance. 

\begin{figure}[htb] 
\centering
\subfigure[]{\includegraphics[width=0.47\textwidth]{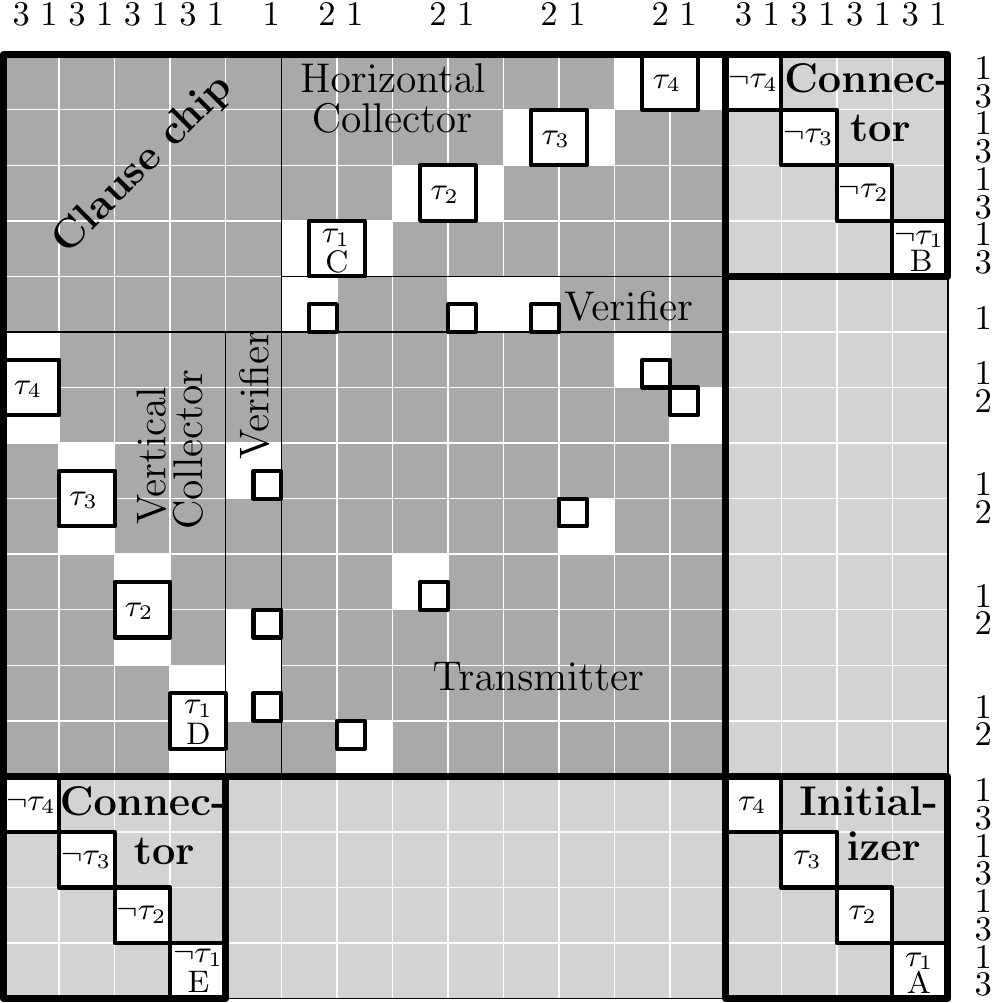}}\hfill
\subfigure[]{\includegraphics[width=0.47\textwidth]{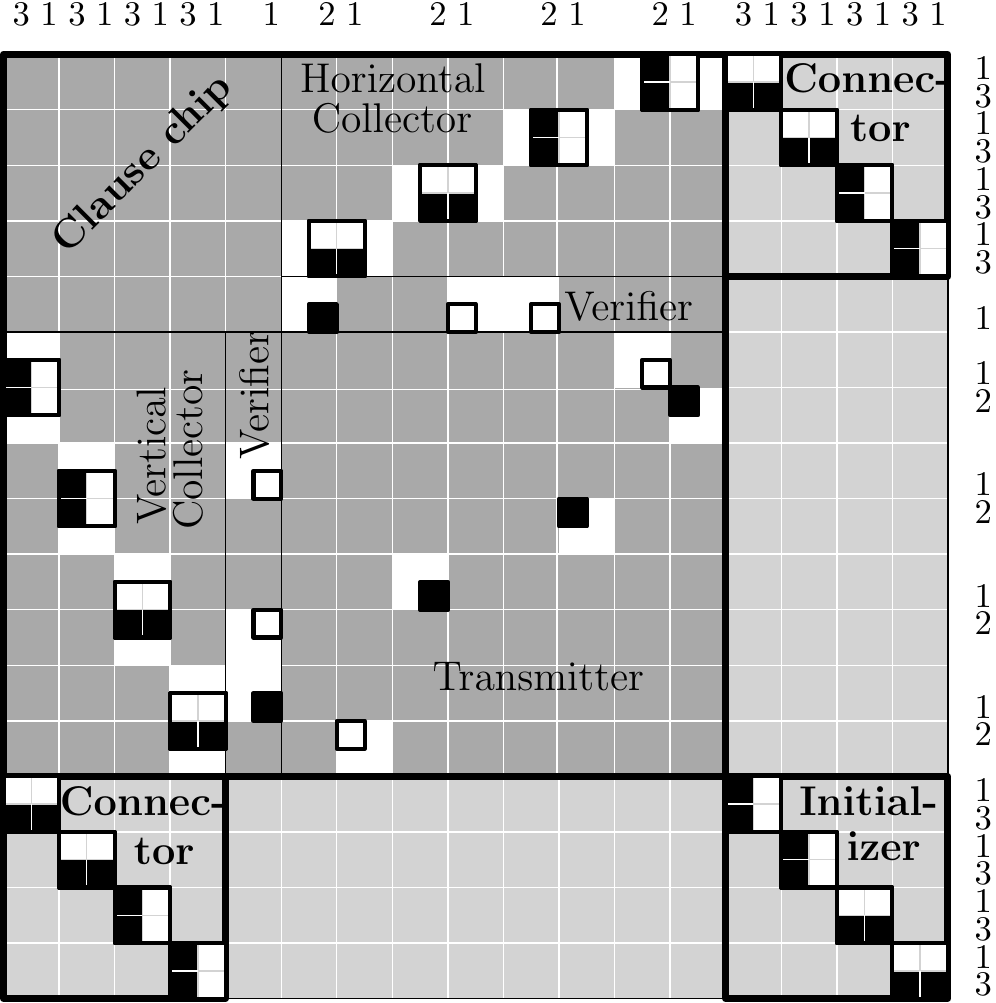}}
\caption{Transformation from \textsc{1-In-3-SAT} for the instance $\mathcal{I}:=(S,T,\mathcal{C}):=(1,4,\{\tau_1\vee \lnot\tau_2 \vee \tau_3\}).$  (a) The circuit board. By setting to zero suitable blocks and row and column sums we make sure that the non-zero components $\xi^*_{p,q}$ of a solution are only possible in the bold-framed boxes within the white blocks in the clause chip, connectors, and the initializer.  (b) A  solution $x^*$ (non-zero components~$\xi^*_{p,q}$ are depicted as black pixels), representing the solution $(\tau^*_1,\tau^*_2,\tau^*_3,\tau^*_4)=(\textsc{True},\textsc{True},\textsc{False},\textsc{False})$ of~$\mathcal{I}.$ }\label{fig:P2example}
\end{figure}

The gray areas indicate blocks that are set to zero via block constraints. By setting to zero suitable row and column sums we make sure that non-zero components $\xi^*_{p,q}$ are only possible in the white bold-framed boxes. The non-zero row and column sums are shown on the top and on the right of the figure. 

The bold-framed $2\times 2$ boxes are \emph{Boolean chips}, which are, according to their position on the circuit board, further classified into~$\tau_t$ and~$\lnot\tau_t$-chips, $t\in [T].$  In these chips we will encode the respective Boolean values for the variables $\tau_t$ and~$\lnot\tau_t$, $t\in [T].$ All blocks, except for those set to zero, are allowed to contain at most two ones. The blocks containing the $\lnot\tau_t$-chips are required to contain precisely two ones. 

In each clause chip there are two strips playing a special role. One is a horizontal strip containing a so-called \emph{horizontal verifier}, the other is a vertical strip containing a so-called \emph{vertical verifier} (see rows~25 and~26, and columns~9 and~10 in Fig.~\ref{fig:P2example}(a)). A key property of our construction is that the truth assignments will be in 1-to-1 correspondence with the solutions of the constructed \textsc{nDR}$(\varepsilon)$ instance if the row and column sum constraints related to the verifiers are left unspecified. The row and column sum constraints related to the verifiers will then ensure that the truth assignments are in fact satisfying truth assignments. 

We can already begin to see how a truth assignment for our particular instance $\mathcal{I}$ will provide a solution for $\mathcal{I}'.$ Utilizing the type~1 and type~2 blocks shown in Fig.~\ref{fig:P2types}, we set each $\tau_t$-chip, $t\in[T],$ of the initializer to type~1 if $\tau_t^*=\textsc{True},$ otherwise we set it to type~2. 

\begin{figure}[htb] 
\centering
\includegraphics[width=0.4\textwidth]{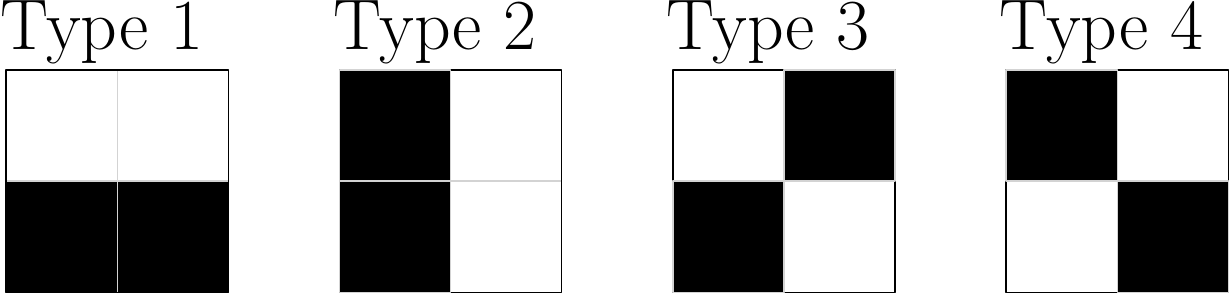}\vspace*{1ex}
\caption{The four possible types of blocks for $\tau_t$- and $\lnot \tau_t$-chips, $t\in[T].$  In $\tau_t$-chips, the types~1 and~2 represent the Boolean values \textsc{True} and \textsc{False}, respectively. (Note that the block types~1-4 correspond to the previously introduced types~$\mathfrak{B}_1,$ $\mathfrak{B}_{3,1},$ $\mathfrak{B}_{3,3},$ and $\mathfrak{B}_{3,4},$ respectively. The new labeling is more concise in the present context as no grouping of different blocks into a single type is required anymore.) }\label{fig:P2types}
\end{figure}

Let us, as an example, consider $t=1$ and the sequence of Boolean chips marked A, B, C, D, and E in Fig.~\ref{fig:P2example}(a); see also Fig.~\ref{fig:transmission}. The chips~A and~B lie in the same vertical strip, and since~A is of type~1,~B needs to be of type~2 in order to satisfy the column sums in the vertical strip. Chip~C must be of type~1 to satisfy the row sums in the horizontal strip. Then, the remaining points in the two vertical strips intersecting chip~C are uniquely determined by the column sums. In particular, the point of the transmitter in a row intersecting~D must be 0. Now, chip~D needs to contain two ones since the vertical strip containing chip~D contains four ones, two of which need to be contained in chip~E via block constraints. Hence, we conclude from the row sums that~D is of type~1. Considering now the vertical strip intersecting~D we see that chip~E is of type~2, and therefore, together with chip~A, satisfies the row sums in the horizontal strip. Thus, the $\tau_1$-chips are consistently of type~1, and the $\lnot\tau_1$-chips are of type~2. Note that we have satisfied all row and column sums for the strips intersecting one of the chips A,B,C,D, and E. Moreover, the block constraints in these chips are also satisfied. 

\begin{figure}[htb] 
\centering
\includegraphics[width=0.45\textwidth]{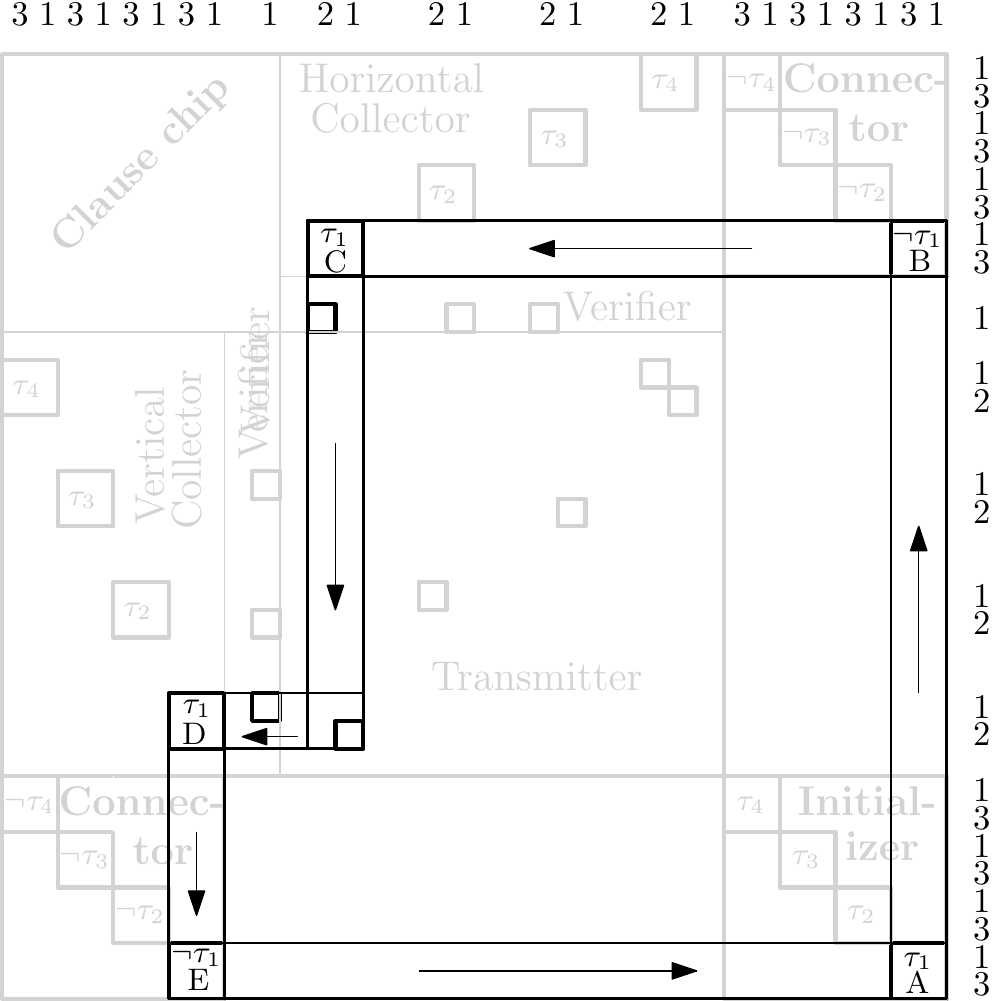}\vspace*{1ex}
\caption{The chain of reasoning (vertical transmission) for deducing the types of the Boolean chips marked A, B, C, D, and E from Fig.~\ref{fig:P2example}. }\label{fig:transmission}
\end{figure}

Turning to the Boolean chips for $t=2,$ we conclude with a similar reasoning that the $\tau_2$- and $\lnot\tau_2$-chips are consistently of type~1 and type~2, respectively.  We have now satisfied the block constraints for the $\tau_2$- and $\lnot\tau_2$-chips, Additionally, we satisfied the row and column sums for the strips intersecting one of these chips. Clearly, we have not violated any of the previously satisfied constraints as there is no row or column that intersects both a Boolean chip for $t=1$ and $t=2.$ 

We say that we have deduced the types of the Boolean chips for~$t\in\{1,2\}$ by \emph{vertical transmission}, because in our previous chain of arguments we started by considering the vertical strip containing the respective Boolean chip of the initializer (see also Fig.~\ref{fig:transmission}). 

Note that such a vertical transmission does not directly work for the assignment \textsc{False} of the variables~$\tau_3,$ $\tau_4$ since the column sums do not allow us to deduce a unique solution for the respective $\lnot\tau_t$-chip. Hence for the Boolean chips for $t\in\{3,4\}$ we resort to \emph{horizontal transmission} to deduce the types of the respective $\tau_t$- and $\lnot\tau_t$-chips. In this way we conclude that the $\tau_t$- and $\lnot\tau_t$-chips are consistently type~2 and type~1, respectively.

Figure~\ref{fig:P2example}(b) shows the solution that we obtain by the previous arguments. All constraints, in particular also the row and column sums in the verifiers, are satisfied. We will later see in general that there is a single one in each of the two verifiers of the $s$-clause chip, $s\in S,$ if, and only if, exactly one of the three literals appearing in the $s$-th clause is \textsc{True}. 

Now, consider the reverse direction for the proof, i.e., suppose a solution $x^*$ to our instance $\mathcal{I}'$ of Fig.~\ref{fig:P2example}(a) is given. Note that if we can ensure that the types of the Boolean chips in the initializer are either~1 or~2 then by the previous arguments we have a satisfying truth assignment for $\mathcal{I}$ by setting $\tau_t^*:=\textsc{True},$ $t\in [T],$ if, and only if, the corresponding $\tau_t$-chip of the initializer is of type~1.

The respective block, row and column sum constraints allow only the four possible types of blocks for the Boolean chips of the initializer shown in Fig.~\ref{fig:P2types}. The types~3 and~4, however, are ruled out by the following ``global'' argument: Suppose, for some~$t\in[T],$ the $\tau_t$-chip of the initializer is of type~$\ell\in\{3,4\}.$ By horizontal transmission, we then conclude that the $\lnot\tau_t$-chips need to be of type~1 (by transmitting, as before in a unique way, through each clause chip). By vertical transmission, however, we conclude that the $\lnot\tau_t$-chips are of type~2, which is a contradiction. 

We remark that effectively we utilize the noisy block constraints at three different places in our construction. They are needed, because we can neither prescribe the exact number of ones in the $\tau_t$-chips of the horizontal or the vertical collectors, nor in the $(s,t)$-configurations that are introduced later and which involve the transmitters (see also Figs.~\ref{fig:P2collectors},~\ref{fig:P2collectorfilled1}, and~\ref{fig:P2collectorfilledX}).

Having illustrated the key ideas behind the construction, we now turn to the formal proof. 

\begin{proof}[Proof of Thm.~\ref{thm:main2}]
Let in the following $\mathcal{I}:=(S,T,\mathcal{C})$ denote an instance of \textsc{1-In-3-SAT}. We will define an instance $\mathcal{I}'$ of \textsc{nDR}$(\varepsilon)$ such that there is a solution for $\mathcal{I}'$ if, and only if, there is one for $\mathcal{I}.$ The circuit board will be contained in the box
\[
[m]\times[n]:=[S(6T+2)+2T]^2.
\]
With each $(p,q)\in [m]\times[n]$ we associate a variable $\xi_{p,q}.$ Figure~\ref{fig:P2layout} illustrates the general layout of the construction. 

\begin{figure}[htb] 
\centering
\includegraphics[width=0.43\textwidth]{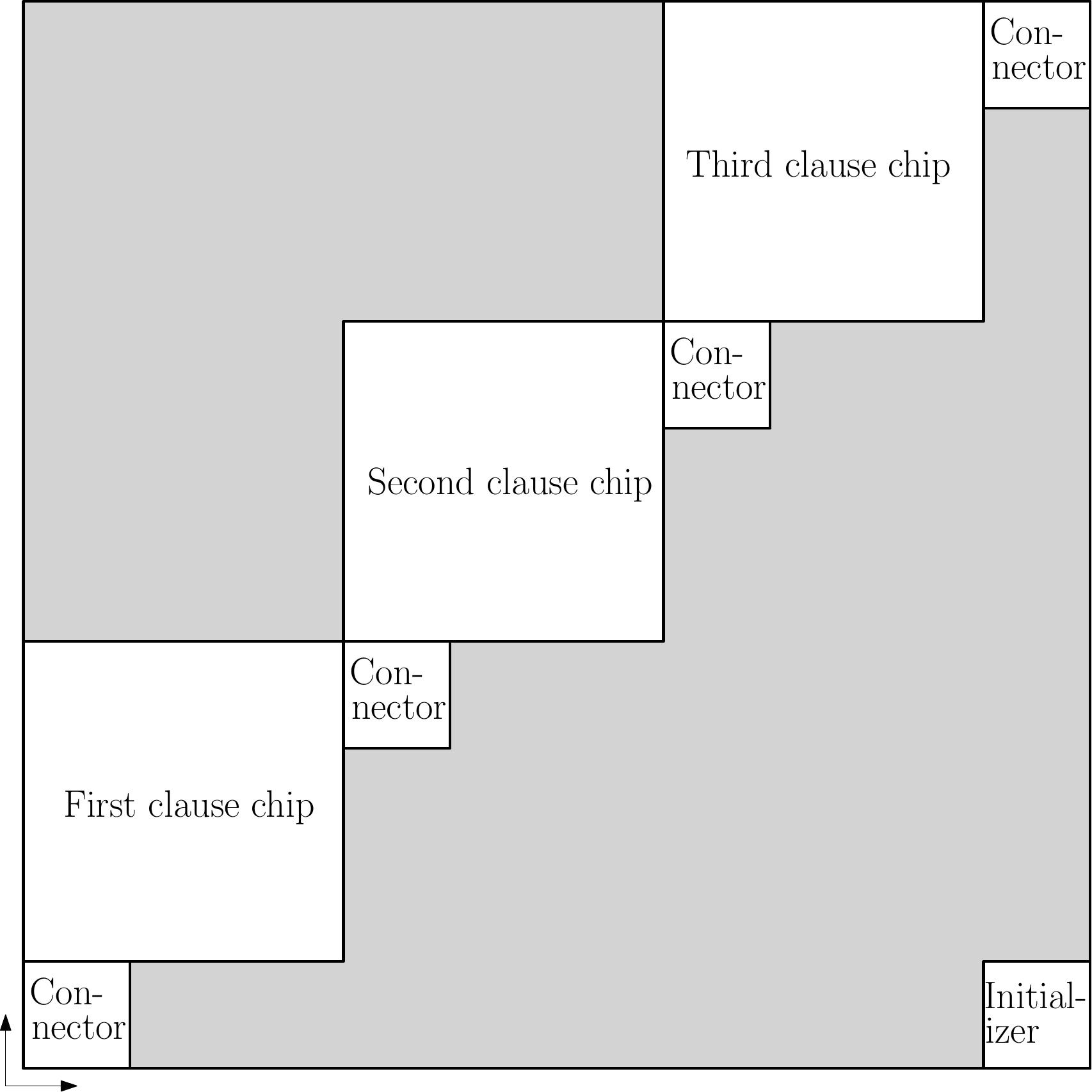}
\caption{The general layout of the circuit board (here for $S=3$). Non-zero components $\xi^*_{p,q}$ of a solution will be possible at most for~$(p,q)$ contained in the white area.}\label{fig:P2layout}
\end{figure}

We define the instance $\mathcal{I'}$ of \textsc{nDR}$(\varepsilon)$ by specifying the row sums $r_1,\dots,r_n,$ the column sums $c_1,\dots,c_m,$ and the block constraints. For notational convenience we define the block constraints via a function $f:C(m,n,2) \to \{(=,0),(=,2),(\approx, 1)\}.$ In this notation, $f(i,j)=(=,\nu),$ $\nu\in\{0,2\},$ denotes the block constraint  
\[
\sum_{(p,q)\in B(i,j)}\xi_{p,q}= \nu,
\]  while $f(i,j)=(\approx,1)$ signifies the block constraint 
\[
\sum_{(p,q)\in B(i,j)}\xi_{p,q}\in 1+([-\varepsilon,\varepsilon]\cap \mathbb{Z}).
\] The set $R$ of reliable block constraints for $\mathcal{I}'$ is then \[R:=\left\{(i,j)\in C(m,n,2)\::\: f(i,j)\neq(\approx,1)\right\}.\]

The different components of the circuit board are placed at specific positions. For a more compact definition of these positions we set \[a_s:=(6T+2)(s-1)+1, \qquad s\in[S+1].\]

Turning to the definition of the initializer, we remark that the initializer is contained inside the box $(a_{S+1},1)+[2T-1]_0^2.$ We will specify the row and column sums when the corresponding connectors are introduced. In terms of the block constraints we define for every $(u,v)\in[T-1]_0^2:$ 
\[
f(a_{S+1}+2u,1+2v) := \left\{\begin{array}{lll} (\approx,1) &:& u+v=T-1,\\ (=,0) &:& \textnormal{otherwise}.\end{array}\right. 
\]
The block $B(a_{S+1}+2(T-t),2t-1),$ $t\in[T],$ will be called  \emph{$\tau_{t}$-chip (of the initializer)}. For an illustration of an initializer, see the bottom right $[8]\times[8]$-box in Fig.~\ref{fig:P2example}(a).

We define $S+1$ connectors. Each connector is contained inside a box $(a_s,a_s)+[2T-1]_0^2$ with $s\in[S+1].$ In terms of block constrains we define for every $s\in [S+1]$ and $(u,v)\in [T-1]_0^2:$
\[
f(a_s+2u,a_s+2v) := \left\{\begin{array}{lll} (=,2) &:& u+v=T-1,\\ (=,0) &:& \textnormal{otherwise}.\end{array}\right. 
\] 
For the row and column sums we define for every $s\in [S+1]$ and $l\in[2T-1]_0:$ 
\[
r_{a_s+l} := c_{a_s+l}:= \left\{\begin{array}{lll} 3 &:& l \in 2\mathbb{N}_0,\\ 1 &:& \textnormal{otherwise}.\end{array}\right. 
\] 
The block $B(a_s+2(T-t),a_s+2(t-1)),$ $s\in [S+1],$ $t\in[T],$ will be called  \emph{$\lnot\tau_{t}$-chip (of the $s$-th connector)}.   For an illustration of a connector, see the bottom left $[8]\times[8]$-box in Fig.~\ref{fig:P2example}(a).

Next we define the $S$ clause chips. Each clause chip is contained in a box $(a_s,a_s+2T)+[6T+1]_0^2$ with $s\in[S].$ 
%\[
%a\in A_2:=A_1\setminus\{S(6T+2)+1\}.
%\] 
The box $(a_s,a_s+2T)+[2T-1]_0\times[4T-1]_0$ is a \emph{(vertical) collector}, the box $(a_s+2T,a_s+2T)+[1]_0\times[4T-1]_0$ is a \emph{(vertical) verifier}, the box $(a_s+2T+2,a_s+2T)+[4T-1]_0^2$ is a \emph{transmitter}, the box $(a_s+2T+2,a_s+6T)+[4T-1]_0\times[1]_0$ is a \emph{(horizontal) verifier}, and the box $(a_s+2T+2,a_s+6T+2)+[4T-1]_0\times[2T-1]_0$ is a \emph{(horizontal) collector}.   For an illustration of a clause chip, see the top left $[26]\times[26]$-box in Fig.~\ref{fig:P2example}(a) depicted in dark gray.

For $s\in [S]$ we define the row sums of the horizontal and the column sums of the vertical verifiers (we will call these row and column sums \emph{verifier sums}) by
\begin{alignat*}{5}
&r_{a_s+6T} &&:=\:\:&& c_{a_s+2T+1} &&:=1,\\
&r_{a_s+6T+1} &&:=&& c_{a_s+2T} &&:=0.
\end{alignat*}
The remaining row and columns sums for the transmitters (meeting also certain verifiers and collectors) are for $s\in [S]$ and $l\in[T-1]_0$ defined by
\begin{alignat*}{5}
&r_{a_s+2T+4l}  &&:=\:\:&&  c_{a_s+2T+4l+2} &&:=0,\\
&r_{a_s+2T+4l+1}&&:=\:\:&&  c_{a_s+2T+4l+3} &&:=2,\\
&r_{a_s+2T+4l+2}&&:=\:\:&&  c_{a_s+2T+4l+4} &&:=1,\\
&r_{a_s+2T+4l+3}&&:=\:\:&&  c_{a_s+2T+4l+5} &&:=0.
\end{alignat*}  
Note that sums for the rows and columns meeting the collectors are already defined as there is a transmitter or connector for each row and column meeting a collector (see also Fig.~\ref{fig:P2layout}). We have thus completed our definition of the row and column sums of our instance and continue now with the remaining block constraints.

For the vertical collector, we define for every $s\in [S]$ and $t\in[T],$ $v_t\in(2\mathbb{Z}+1)\cap\{3-4t,3-4t+1,\dots,4(T-t)+1\}:$
\[
f(a_s+2(T-t),a_s+2T+4t-3+v_t):=\left\{\begin{array}{lll}(\approx,1) &:& v_t\in\{-1,1\},\\(=,0)&:& \textnormal{otherwise}.\end{array}\right.
\]
The box $(a_s+2(T-t),a_s+2T+4t-3)+[1]_0\times[1]_0,$ $t\in[T],$ will be called \emph{$\tau_t$-chip (of the $s$-th vertical collector)}. 

Similarly, for the horizontal collector we define for every $s\in [S],$ $t\in[T],$ $u_t\in(2\mathbb{Z}+1)\cap\{3-4t,3-4t+1,\dots,4(T-t)+1\}:$ 
\[
f(a_s+2T+4t-1+u_t,a_s+6T+2t):=\left\{\begin{array}{lll}(\approx,1) &:& u_t\in \{-1,1\},\\(=,0)&:& \textnormal{otherwise}.\end{array}\right.
\] The box $(a_s+2T+4t-1,a_s+6T+2t)+[1]_0\times[1]_0,$ $t\in[T],$ will be called \emph{$\tau_t$-chip (of the $s$-th horizontal collector)}. 

Note that the $\tau_t$-chips, $t\in[T],$ of the collectors are not contained in blocks. In fact, each such $\tau_t$-chip intersects two blocks (which, in turn, are allowed to contain at most two ones). For any $t\in[T]$ we will also refer to the $\tau_t$- and $\lnot\tau_t$-chips on the circuit board as \emph{Boolean chips.}  

For the empty space above the vertical collector, i.e., for $s\in [S]$ and $(u,v)\in[T-1]_0^2,$  we define
\[f(a_s+2u,a_s+6T+2v):=(=,0).\]

Now we turn to the block constraints for the clause-dependent part of chip, i.e., the verifiers and the transmitters.
  
With $U_s$ and $N_s,$ $s\in [S],$ we denote the indices of the variables that appear unnegated and, respectively, negated in the $s$-th clause. For instance, for $\tau_1\vee\lnot\tau_2\vee\tau_3$ (regarded as the first clause) we have $U_1=\{1,3\}$ and $N_1=\{2\}.$ 

Depending on whether a given variable $\tau_t$ appears negated, unnegated or not at all in the $s$-th clause, we will distinguish three different types of verifier and transmission parts forming a so-called $(s,t)$-configuration. Formally, for $s\in [S]$ and $t \in [T]$  we refer to the sets of points
\[
(a_{s}+2T,a_{s}+2T)+
\left\{\begin{array}{lll}  \{(1,4t-2), (4t,4t-3), (4t-1,4T)\} &:& t \in U_{s},\\
\{(1,4t-3),(4t-1,4t-2),(4t,4T)\} &:& t\in N_{s},\\
\{(4t-1,4t-2),(4t,4t-3)\} &:& \textnormal{otherwise,}\\\end{array}\right.
\] as \emph{$(s,t)$-configuration}. Figure~\ref{fig:P2collectors} gives an illustration.

\begin{figure}[htb] 
\centering
\subfigure[]{\fbox{\includegraphics[width=.16\textwidth]{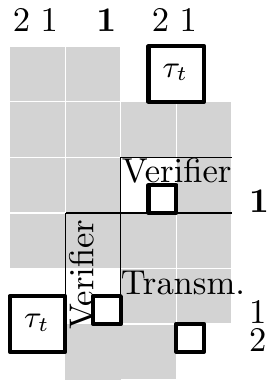}}}\hspace*{4ex}
\subfigure[]{\fbox{\includegraphics[width=.16\textwidth]{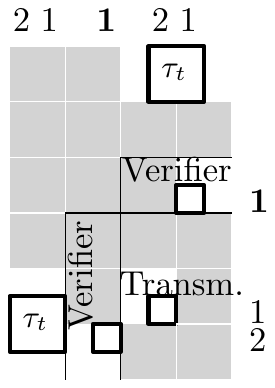}}}\hspace*{4ex}
\subfigure[]{\fbox{\includegraphics[width=.16\textwidth]{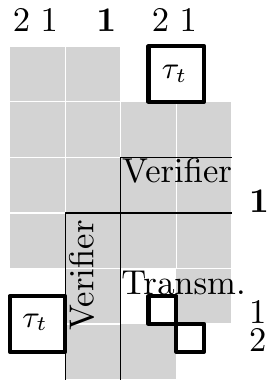}}}
\caption{The 3 possible $(s,t)$-configurations (depicted as bold-framed $1\times1$ boxes) for (a)~$t\in U_s,$ (b)~$t\in N_s,$ and (c)~$t\in[T]\setminus(U_s\cup N_s)$.}\label{fig:P2collectors}
\end{figure}

For the three different types of $(s,t)$-configurations, we need to set different blocks to zero. The formal definition is as follows.

For the vertical verifiers we define for every $s\in[S],$ $t\in [T], $ and $v\in[2]:$
{\small
\[
f(a_s+2T,a_s+2T+4t-2v):=\left\{\hspace*{-0.5ex}\begin{array}{lll}(\approx,1) &\hspace*{-1ex}:\hspace*{-1ex}& ((t\in U_s)\wedge (v=1))\vee((t
\in N_s)\wedge (v=2)),\\(=,0)&\hspace*{-1ex}:\hspace*{-1ex}& \textnormal{otherwise}.\end{array}\right.
\]}

Similarly, for the horizontal verifiers we define for every $s\in[S],$ $t\in [T],$ and $u\in[2]:$
{\small
\[
f(a_s+2T+2+4t-2u,a_s+6T):=\left\{\hspace*{-0.5ex}\begin{array}{lll}(\approx,1) &\hspace*{-1.7ex}:\hspace*{-1.7ex}& ((t\in U_s)\wedge (u=2))\vee((t
\in N_s)\wedge (u=1)),\\(=,0)&\hspace*{-1.7ex}:\hspace*{-1.7ex}& \textnormal{otherwise}.\end{array}\right.
\]}

For the transmitters we define for every $s\in[S]$ and $u\in [2T-1]:$  
{\small
\begin{align*}
f(a_s+2T+2+2u,a_s+2T+2(u-1)) &:=\left\{\hspace*{-1.5ex}\begin{array}{lll}(\approx,1) &\hspace*{-1.7ex}:\hspace*{-1.7ex}& (u\in2\mathbb{N}_0+1)\wedge((u+1)/2\in U_{s}),\\(=,0)&\hspace*{-1.7ex}:\hspace*{-1.7ex}& \textnormal{otherwise},\end{array}\right.\\
f(a_s+2T+2+2(u-1),a_s+2T+2u) &:=\left\{\hspace*{-1.5ex}\begin{array}{lll}(\approx,1) &\hspace*{-1.7ex}:\hspace*{-1.7ex}& (u\in2\mathbb{N}_0+1)\wedge ((u+1)/2\in N_{s}),\\(=,0)&\hspace*{-1.7ex}:\hspace*{-1.7ex}& \textnormal{otherwise}.\end{array}\right.
\end{align*}} The remaining block constraints in the transmitter are set to zero, i.e., for every $s\in[S],$ $(u,v)\in[2T-1]_0^2$ with $u-v\not\in\{-1,1\}$ we set
\[f(a_s+2T+2+2u,a_s+2T+2v) :=(=,0).\]

Outside the initializer, connectors, and clause chips we set everything to zero, i.e., for every $(i,j) \in ([m]\times[n])\cap C(m,n,2)$ with
\begin{multline}
(i,j)\not\in
((a_{S+1},1)+[2(T-1)]_0^2)\cup((a_{S+1},a_{S+1})+[2(T-1)]_0^2)\\ \cup\bigcup_{s\in[S]}\left(((a_s,a_s)+[2(T-1)]_0^2)\cup ((a_s,a_s+2T)+[6T+1]_0^2))\right)
\end{multline} we set $ f(i,j):=(=,0).$

This concludes the formal definition of the instance $\mathcal{I}'.$ We shall now show that~$\mathcal{I}'$ admits a solution if, and only if, the \textsc{1-In-3-SAT} instance $\mathcal{I}$  admits a solution.

Let $\mathcal{G}\subseteq [m]\times[n]$  denote the set of points of the Boolean chips and $(s,t)$-configurations, $(s,t)\in [S]\times[T].$ 

\textbf{Claim 1:} In any solution $x^*$ of $\mathcal{I}'$ we have $\xi^*_{p,q}=0$ for every $(p,q)\in ([m]\times[n])\setminus \mathcal{G}.$ 

To see this, we first remark that the claim holds for all points outside each clause chip since the only blocks there that are not set to zero are  the Boolean chips of the initializer and connectors. For the points inside a clause chip, we need to distinguish three cases since there are three differently structured $(s,t)$-configurations, in which different blocks are set to zero. They are shown in Fig.~\ref{fig:P2collectors}. In all three cases the points outside the corresponding Boolean chips in the two collectors and the $(s,t)$-configuration are either set to zero by block constraints or by zeros of suitable row or column sums. This shows the claim.

\textbf{Claim 2:} In any solution $x^*$ of $\mathcal{I}'$ every $\tau_t$- and $\lnot\tau_t$-chip, $t\in[T],$ is of one of the types from~$\{1,2,3,4\}$ shown in Fig.~\ref{fig:P2types}.

For this, consider a fixed index $t \in [T].$ Each $\tau_t$-chip is contained in a horizontal or vertical strip that, except for a $\lnot\tau_t$-chip, contains no other points of $\mathcal{G}.$ Since by the block constraints each $\lnot\tau_t$-chip contains two ones, and since the respective row or column sums in the strip sum up to four there need to be exactly two ones also in each $\tau_t$-chip.  The row and column sums enforce that no two ones are in the rightmost column and upper row of the box of each $\tau_t$-chip. Hence, there are at most the four possibilities shown in Fig.~\ref{fig:P2types}.

\textbf{Claim 3:} In any solution $x^*$ of $\mathcal{I}'$ all the $\tau_t$-chips (for a given $t\in[T]$) are either of type~1 or type~2.

To show this, consider again a fixed index $t \in [T].$ Let $V_0,\dots, V_{4S+2}\subseteq \mathcal{G}$ denote the sequence with (i)~$V_0=V_{4S+2}$ is the $\tau_t$-chip of the initializer, (ii)~$V_1$ is the $\lnot\tau_t$-chip of the $(S+1)$-th connector, (iii)~$V_{4s+1}$ is the $\lnot\tau_t$-chip of the $(S-s+1)$-th connector, (iv)~$V_{4s-2}$ is the $\tau_t$-chip of the $(S-s+1)$-th horizontal collector, (v)~$V_{4s-1}$ is the set of points in the transmitter of the $(s,t)$-configuration of the $(S-s+1)$-th clause, and (vi)~$V_{4s}$ is the $\tau_t$-chip of the $(S-s+1)$-th vertical collector, $s\in [S];$ see Fig.~\ref{fig:transmission} for the case $S=1$ and Fig.~\ref{fig:P2layout} again for the general layout.

By horizontal transmission (see Fig.~\ref{fig:P2collectorfilledX}) and Claim~2, we have the following implication: If a $\tau_t$-chip of $V_{2l},$ $l\in [2S]_0,$ is of type~2, 3 or $4,$ then the $\tau_t$-chip of $V_{2l+2}$ is of type~2, and the points of $V_{2l+1}$ are also uniquely determined.

By vertical transmission (see Fig.~\ref{fig:P2collectorfilled1}) and Claim~2, we have the following implication:
If a $\tau_t$-chip of $V_{2l},$ $l\in [2S+1],$ is of type~1, 3 or $4,$ then the $\tau_t$-chip of $V_{2l-2}$ is of type~1, and the points of $V_{2l-1}$ are also uniquely determined.

\begin{figure}[htb] 
\centering
\includegraphics[width=0.9\textwidth]{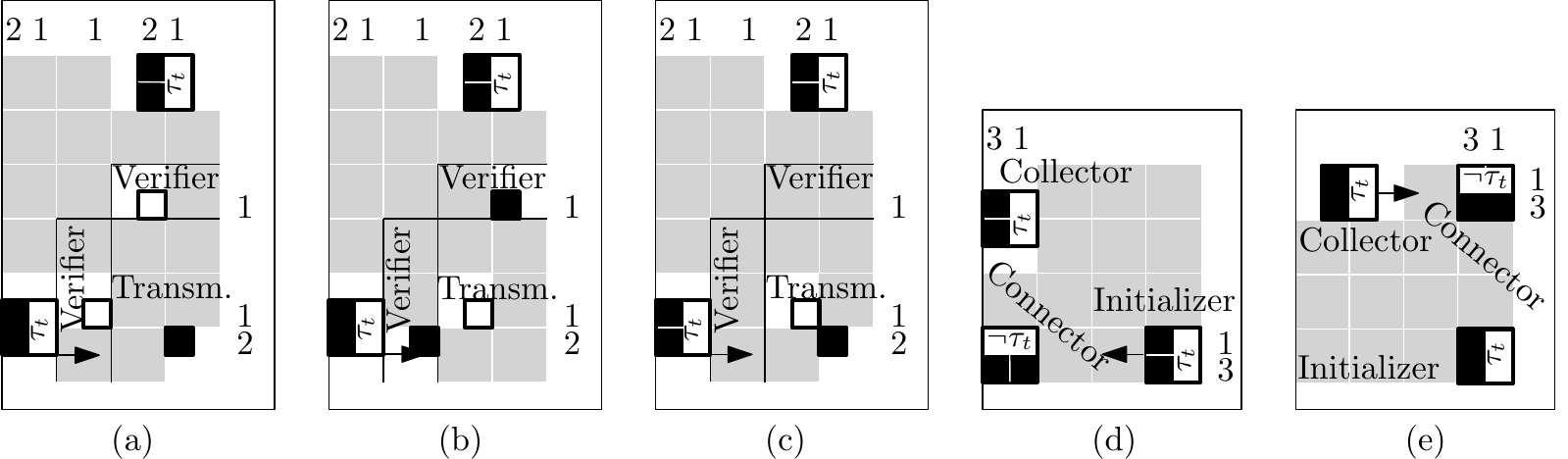}\vspace*{1ex}
\caption{Horizontal transmission of a type~2 block. (a,b,c) Transmission through the three differently structured $(s,t)$-configurations. (d,e) Transmission through a connector. Type~3 and~4 blocks are transmitted in the same way (only the initial block, indicated by the arrow tail, is replace by the respective type~3 or 4 block).}\label{fig:P2collectorfilledX}
\end{figure}

\begin{figure}[htb] 
\centering
\includegraphics[width=0.9\textwidth]{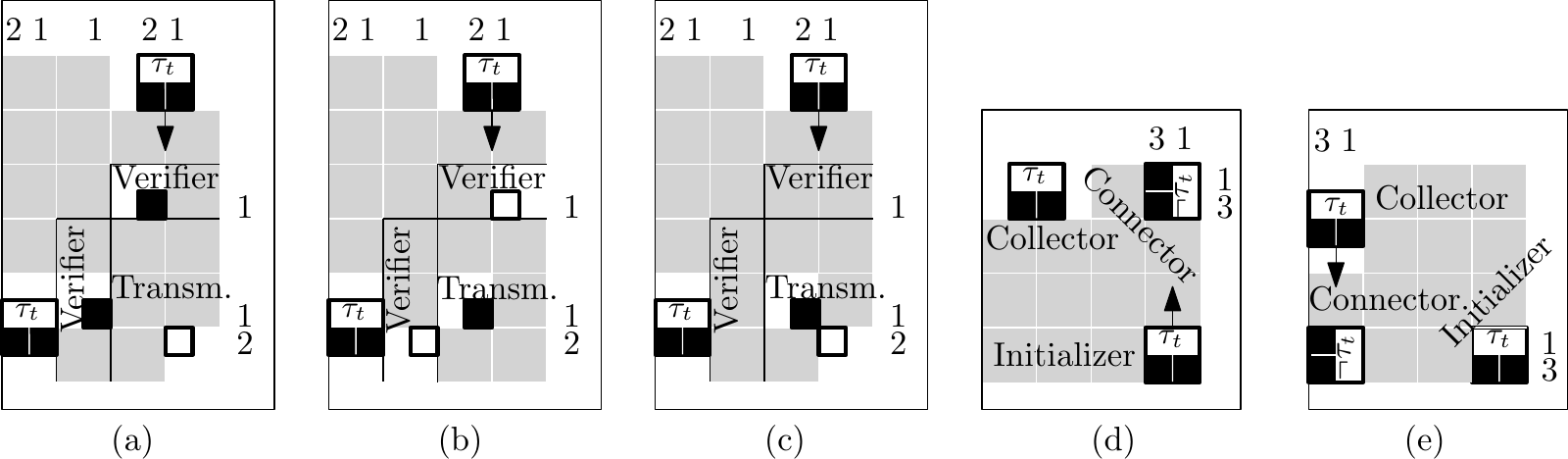}\vspace*{1ex}
\caption{Vertical transmission of a type~1 block. (a,b,c) Transmission through the three differently structured $(s,t)$-configurations. (d,e) Transmission through a connector. Type~3 and~4 blocks are transmitted in the same way (only the initial block, indicated by the arrow tail, is replace by the respective type~3 or 4 block).}\label{fig:P2collectorfilled1}
\end{figure}

Hence, if we start at an arbitrary $\tau_t$-chip and perform vertical and horizontal transmissions that meet at some Boolean chip we obtain a contradiction unless the $\tau_t$-chip was of type~1 or~2.

The types~1 and~2 can be associated with Boolean values. In $\tau_t$-chips, and in particular those of the initializer, we associate with type~1 the value \textsc{True} and with type~2 the value \textsc{False}.

\textbf{Claim 4:} There is a 1-to-1 correspondence between the (not-necessarily satisfying) truth-assignments for $\mathcal{I}$ and the solutions~$x^*$ to the $\mathcal{I}''$-variant of $\mathcal{I}'$ that allows violation of the verifier sums, i.e., of the column sum constraints $c_i=1$ for the vertical verifiers or the row sum constraints $r_i=1$ for the horizontal verifiers. 

By construction, there is neither a horizontal nor a vertical strip that intersects the Boolean chips for two variables $\tau_{i_1},$ $\tau_{i_2},$ with $i_1\neq i_2.$ The arguments from Claim~3 are thus independently valid for each $t\in[T].$ We can hence consider an arbitrary (not-necessarily satisfying) truth-assignment for $\mathcal{I},$ fill in the corresponding types in the initializer and extend them to yield a solution to the $\mathcal{I}''$-variant of $\mathcal{I}'$ with unspecified verifier sums. Vice versa, the chip types in the initializer of a solution to $\mathcal{I}''$ yield a truth-assignment for~$\mathcal{I}.$  We have already seen in the proof of Claim~3 that the type of each $\tau_t$-chip, $t\in[T],$ in the initializer determines all ones in the corresponding $\lnot\tau_t$-chips and $(s,t)$-configurations, $s\in[S].$ Hence, the correspondence is indeed 1-to-1.

\textbf{Claim 5:} The satisfying truth-assignments of $\mathcal{I}$ are in 1-to-1 correspondence with the solutions of~$\mathcal{I}'.$ 

This can now be seen by observing that the $\tau_t$-chip, $t\in[T],$ of the $s$-th horizontal and vertical collector, $s\in [S],$ contributes a one to the horizontal and also to the vertical verifier in the $s$-th clause chip if, and only if, the $\tau_t$-chip is \textsc{True} (is of type~1) with $t\in U_s$ or the $\tau_t$-chip is \textsc{False} (is of type~2) with $t\in N_s$; see Figs.~\ref{fig:P2collectors}, \ref{fig:P2collectorfilled1}(a,b,c), and~\ref{fig:P2collectorfilledX}(a,b,c).  The verifier sums thus ensure that exactly one literal in the $s$-th clause is \textsc{True}. This, together with Claim~4, proves the claim.

By Claim~5, $\mathcal{I}'$ admits a solution if, and only if, $\mathcal{I}$ admits a solution.

Finally, we note that the transformation runs in polynomial time.
\end{proof}

In section~\ref{sect:notation} we remarked that $n\textsc{DR}(\varepsilon)$ can be viewed as version of $\textsc{DR}$ where small ``occasional'' uncertainties in the gray levels are allowed. Reviewing the proof of Thm.~\ref{thm:main2}, the term ``occasional'' can be quantified as meaning ``of the order of the square root of the number of blocks.''  More precisely, of $(S(3T+1)+T)^2$ blocks on the circuit board, there are $S(6T+3)+T$ with uncertainty.

Let us further remark that the arguments in the above proof do not rely on the particular value of $\varepsilon>0$ in the noisy block constraints. These constraints can be replaced by any other types of block constraints, certainly as long as they allow for 0, 1, and~2 ones to be contained in each of the corresponding blocks. By this observation it is possible to adapt our proof to several other contexts where the reconstruction task involves other types or combinations of block constraints; see \cite{agwindowconstraints}.

Finally, note that the transformation in the proof of Thm.~\ref{thm:main2} is parsimonious (see Claim~5 in the above proof). Hence the problem of deciding whether a given solution of an instance of \textsc{nDR}$(\varepsilon)$ with $\varepsilon>0$ has a non-unique solution is $\mathbb{N}\mathbb{P}$-complete (which proves Cor.~\ref{cor:uniqueness}); for an $\mathbb{N}\mathbb{P}$-hardness proof of $\textsc{Unique-1-in-3-SAT}$ see \cite[Lem.~4.2]{ggp-99}. 

\begin{proof}[Proof of Cor.~\ref{cor:largebox}] We use a transformation from \textsc{nDR}$(\varepsilon).$ Let \[\mathcal{I}:=(m,n,r_1,\dots,r_n,c_1,\dots,c_m,R,v(1,1),\dots,v(m-1,n-1))\] denote an instance of \textsc{nDR}$(\varepsilon).$ Note that $n$ and $m$ are even.

We set $m':=mk/2$ and $n':=nk/2.$ For $(i,j)\in ([m]\times[n])\cap C(m,n,2)$ and $l\in[k]$ we define
\begin{align*}
r'_{\frac{k}{2}(j-1)+l}&:=\left\{\begin{array}{lll}r_{j+l}&:&l\in[2],\\ 0&:&\textnormal{otherwise,}\end{array}\right.\\
c'_{\frac{k}{2}(i-1)+l}&:=\left\{\begin{array}{lll}c_{i+l}&:&l\in[2],\\ 0&:&\textnormal{otherwise,}\end{array}\right.\\
v'(k(i-1)/2+1,k(j-1)/2+1))&:=v(i,j),
\end{align*} and
\[R':=\{(k(i-1)/2+1,k(j-1)/2+1) : (i,j)\in R\}.\]

This gives the instance \[\mathcal{I'}:=(m',n',r'_1,\dots,r'_{n'},c'_1,\dots,c'_{m'},R',v'(1,1),\dots,v'(m'-k+1,n'-k+1))\] of~\textsc{nSR}$(k,\varepsilon).$ Clearly, the instance~$\mathcal{I}$ of~\textsc{nDR}$(\varepsilon)$ admits a solution if, and only if, the instance~$\mathcal{I}'$ of \textsc{nSR}$(k,\varepsilon)$ admits a solution (by filling/extracting the~$2\times2$-blocks of $\mathcal{I}$ into/from the~$k\times k$-blocks of $\mathcal{I}'$). This is a polynomial-time transformation, proving the corollary.
\end{proof}

\section{Final Remarks}\label{sect:5}
Recall that in the proof of Thm.~\ref{thm:main1} we have made extensive use of the concept of local switches. Each such switch involves four points. 
Figure~\ref{fig:largeswitchingcomp} shows, on the other hand, two solutions of a \textsc{DR} instance that differ by~$20$ $1\times1$-pixels. It is easily checked that these are the only solutions for the given instance (Note that we used such a structure in the proof of Thm.~\ref{thm:main2}). Hence, there is no sequence of small switches transforming one solution into the other. The existence of such large irreducible switches is well-known for the $\mathbb{N}\mathbb{P}$-hard reconstruction problem from more than two directions \cite{kongherman98} (and would follow independently from $\mathbb{P}\neq\mathbb{N}\mathbb{P}$). The problem $\textsc{DR},$ however, admits large switches albeit is polynomial-time solvable. 

\begin{figure}[htb] 
\centering
\subfigure[]{\includegraphics[width=0.20\textwidth]{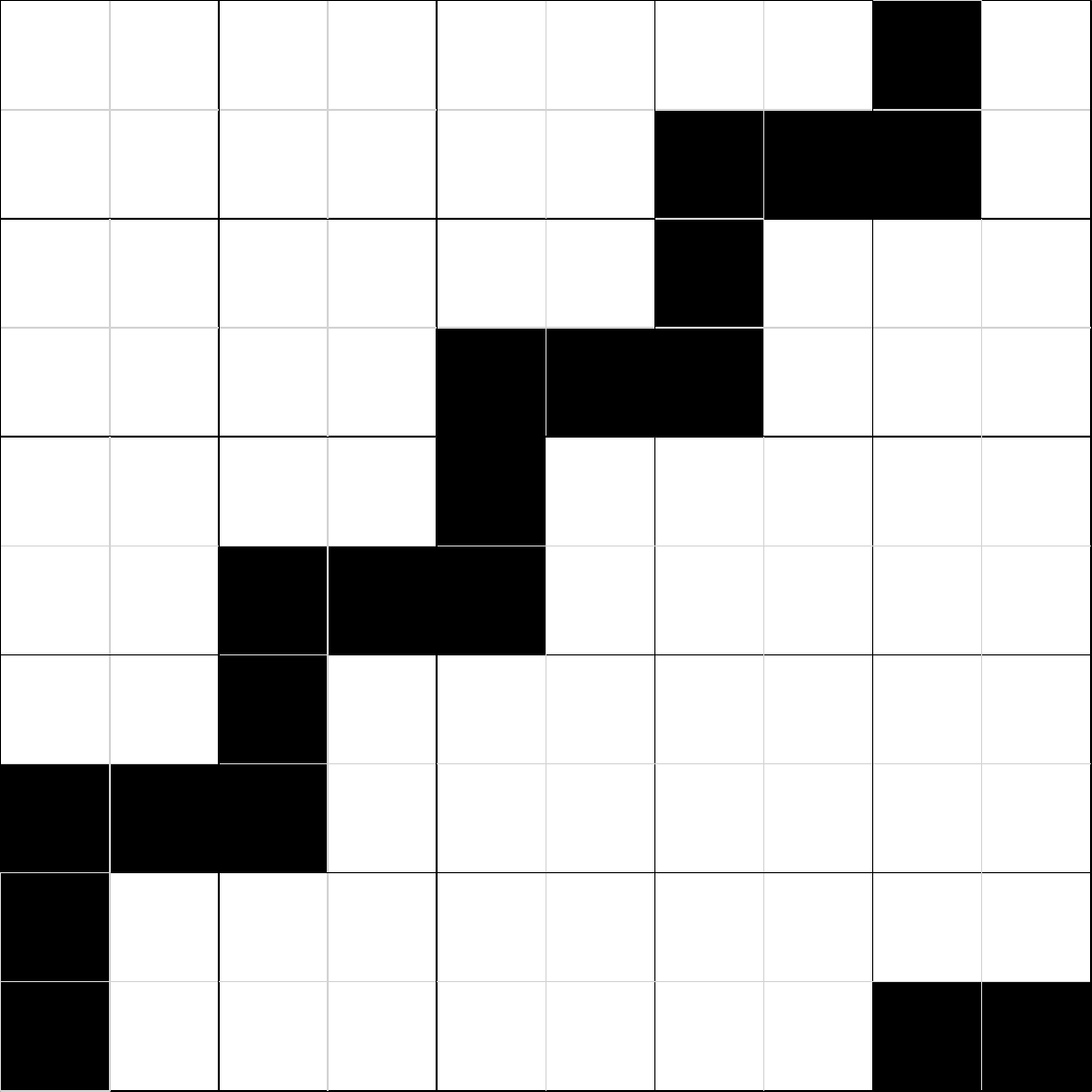}}\hspace*{4ex}
\subfigure[]{\includegraphics[width=0.20\textwidth]{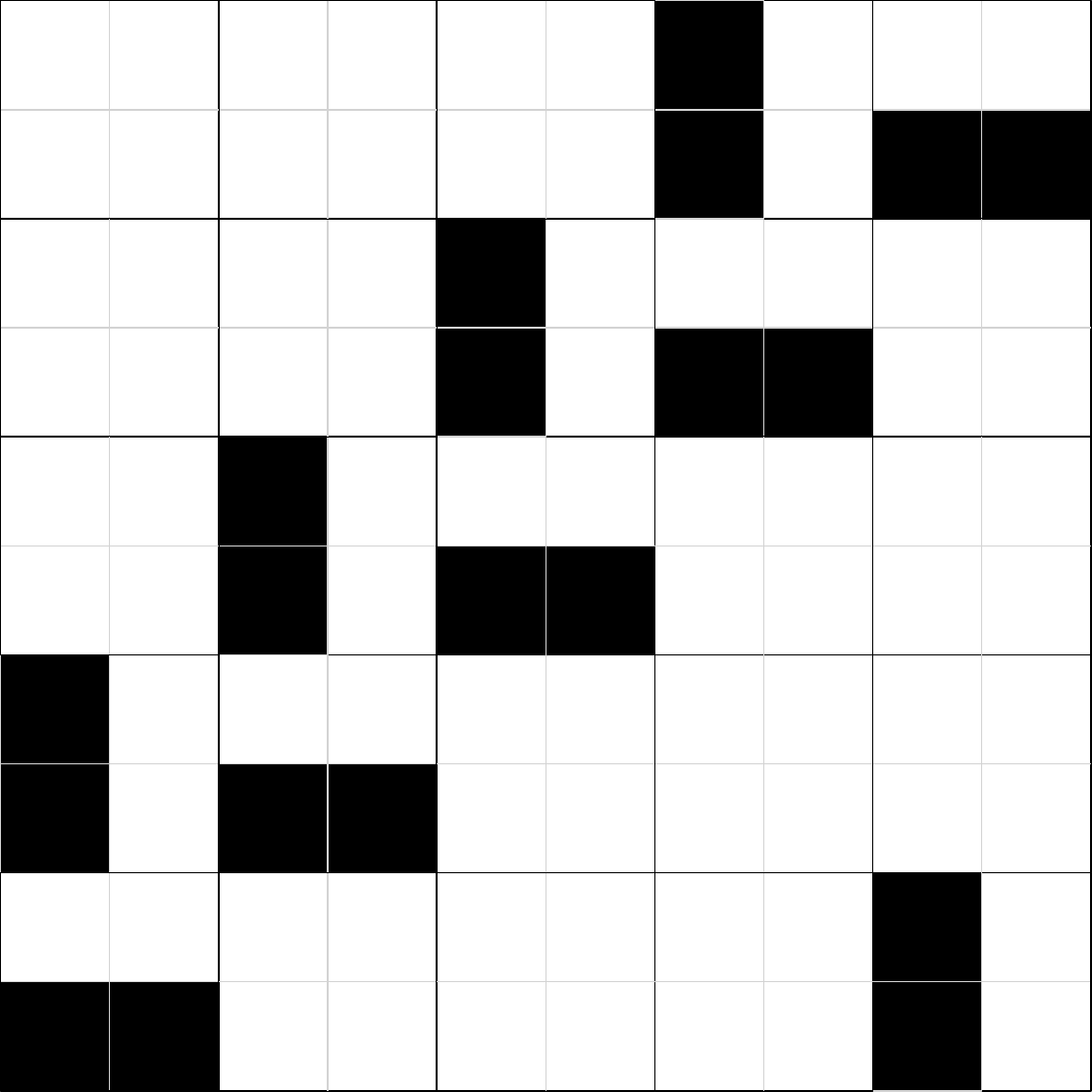}}
\caption{An example of a large switching component (in the sense of~\cite{kongherman98}) for an instance of~\textsc{DR}. The solutions shown in~(a) and (b) are the only two solutions for that instance. Both solutions are reduced.}\label{fig:largeswitchingcomp}
\end{figure}

In addition to the block constraints our computational tasks rely only on row and column sums, i.e., projection data from two viewing angles. For more than two viewing directions we obtain the $\mathbb{N}\mathbb{P}$-hardness of \textsc{nSR}$(k,\varepsilon)$ for any $k\geq2$ and $\varepsilon>0$ directly from the construction in~\cite{ggp-99}. The $\mathbb{N}\mathbb{P}$-hardness of the corresponding variants of \textsc{nSR}$(k,0)$ for $k\geq2,$ can also be derived from~\cite{ggp-99}. First the construction is thinned out so as to make sure that each block contains at most one element of the grid of candidate points. Then, a second copy with inverted colors is interwoven by applying a translation by $(1,1)^T$ or $(1,-1)^T.$ This ensures that each non-empty block contains exactly one one.

The complexity status of \textsc{nSR}$(k,0)$ remains open for~$k\not\in\{1,2\}.$ We conjecture that the problem is $\mathbb{N}\mathbb{P}$-hard for larger~$k.$ 

Also, it would be interesting to analyze the effect of the results of the present paper on the ``degree of ill-posedness'' of the underlying problem (see \cite{ag06}).

\section*{Acknowledgments}
We are grateful to the anonymous referees for their valuable comments on a previous version of the paper. This work was supported by the German Research Foundation Grant GR 993/10-2 and the European COST Network~MP1207.


\begin{thebibliography}{10}

\bibitem{batenburgnature}
S.~Van Aert, K.~J. Batenburg, M.~D. Rossell, R.~Erni, and G.~Van Tendeloo.
\newblock Three-dimensional atomic imaging of crystalline nanoparticles.
\newblock {\em Nature}, 470(7334):374--376, 2011.

\bibitem{philmag}
A.~Alpers, A.~Brieden, P.~Gritzmann, A.~Lyckegaard, and H.~F. Poulsen.
\newblock Generalized balanced power diagrams for {3D} representations of
  polycrystals.
\newblock {\em Phil. Mag.}, 95(9):1016--1028, 2015.

\bibitem{alpersgardner13}
A.~Alpers, R.J. Gardner, S.~K{\"o}nig, R.S. Pennington, C.B. Boothroyd,
  L.~Houben, R.E. Dunin-Borkowski, and K.J. Batenburg.
\newblock Geometric reconstruction methods for electron tomography.
\newblock {\em Ultramicroscopy}, 128(C):42--54, 2013.

\bibitem{ag06}
A.~Alpers and P.~Gritzmann.
\newblock On stability, error correction, and noise compensation in discrete
  tomography.
\newblock {\em SIAM J. Discrete Math.}, 20(1):227--239, 2006.

\bibitem{agdynamic}
A.~Alpers and P.~Gritzmann.
\newblock Dynamic discrete tomography.
\newblock {\em Inverse Probl.}, 34(3):034003 (26pp), 2018.

\bibitem{agwindowconstraints}
A.~Alpers and P.~Gritzmann.
\newblock Reconstructing binary matrices under window constraints from their
  row and column sums.
\newblock {\em Fund. Informaticae}, 155(4):321--340, 2017.

\bibitem{agms-15}
A.~Alpers, P.~Gritzmann, D.~Moseev, and M.~Salewski.
\newblock {3D} particle tracking velocimetry using dynamic discrete tomography.
\newblock {\em Comput. Phys. Commun.}, 187(1):130--136, 2015.

\bibitem{apkh-06}
A.~Alpers, H.~F. Poulsen, E.~Knudsen, and G.~T.Herman.
\newblock A discrete tomography algorithm for improving the quality of {3DXRD}
  grain maps.
\newblock {\em J. Appl. Crystallogr.}, 39(4):582--588, 2006.

\bibitem{bertero}
M.~Bertero and P.~Boccacci.
\newblock {\em Introduction to Inverse Problems in Imaging}.
\newblock IOP Publishing, Philadelphia, 1998.

\bibitem{modal6}
O.~Bill, M.~Faouzi, R.~Meuli, P.~Maeder, M.~Wintermark, and P.~Michel.
\newblock Added value of multimodal computed tomography imaging: analysis of
  1994 acute ischaemic strokes.
\newblock {\em Eur. J. Neurol.}, 24(1):167--174, 2017.

\bibitem{brualdi}
R.~A. Brualdi.
\newblock {\em Combinatorial Matrix Classes}.
\newblock Cambridge University Press, Cambridge, 2006.

\bibitem{chambolle}
A.~Chambolle.
\newblock An algorithm for total variation minimization and applications.
\newblock {\em J. Math. Imaging Vision}, 20(1):89--97, 2004.

\bibitem{cook1971}
S.~A. Cook.
\newblock The complexity of theorem-proving procedures.
\newblock In {\em Proc. 3rd Ann. ACM Symp. Theory of Computing}, pages
  151--158. ACM, 1971.

\bibitem{modal4}
Z.~Di, S.~Leyffer, and S.~M. Wild.
\newblock Optimization-based approach for joint x-ray fluorescence and
  transmission tomographic inversion.
\newblock {\em SIAM J. Imaging Sciences}, 9(1):1--23, 2016.

\bibitem{duerr-Guinez-matamala-12}
C.~D{\"u}rr, F.~Gui{\~n}ez, and M.~Matamala.
\newblock Reconstructing 3-colored grids from horizontal and vertical
  projections is {NP}-hard: {A} solution to the 2-atom problem in discrete
  tomography.
\newblock {\em SIAM J. Discrete Math}, 26(1):330--352, 2012.

\bibitem{gale57}
D.~Gale.
\newblock A theorem on flows in networks.
\newblock {\em Pacific J. Math.}, 7:1073--1082, 1957.

\bibitem{gg97}
R.~J. Gardner and P.~Gritzmann.
\newblock Discrete tomography: {D}etermination of finite sets by x-rays.
\newblock {\em Trans. Amer. Math. Soc.}, 349(6):2271--2295, 1997.

\bibitem{ggp-99}
R.~J. Gardner, P.~Gritzmann, and D.~Prangenberg.
\newblock On the computational complexity of reconstructing lattice sets from
  their {X}-rays.
\newblock {\em Discrete Math.}, 202(1-3):45--71, 1999.

\bibitem{gareyjohnson}
M.~R. Garey and D.~S. Johnson.
\newblock {\em Computers and Intractability}.
\newblock W. H. Freeman and Co., San Francisco, USA, 1979.

\bibitem{superresolution2}
G.~Van Gompel, K.~J. Batenburg, E.~Van de~Casteele, W.~van Aarle, and
  J.~Sijbers.
\newblock A discrete tomography approach for superresolution micro-{CT} images:
  application to bone.
\newblock In {\em 2010 IEEE Intern. Symp. Biomedical Imaging: From Nano to
  Macro}, pages 816--819, 2010.

\bibitem{budget2}
F.~Grandoni, R.~Ravi, M.~Singh, and R.~Zenklusen.
\newblock New approaches to multi-objective optimization.
\newblock {\em Math. Program.}, 146(1):525--554, 2014.

\bibitem{gritzmann97}
P.~Gritzmann.
\newblock On the reconstruction of finite lattice sets from their x-rays.
\newblock In {\em Discrete Geometry for Computer Imagery (Eds.: E.~Ahronovitz
  and C.~Fiorio), DCGI'97, Lecture Notes on Computer Science 1347, Springer},
  pages 19--32, 1997.

\bibitem{glw11}
P.~Gritzmann, B.~Langfeld, and M.~Wiegelmann.
\newblock Uniqueness in discrete tomography: Three remarks and a corollary.
\newblock {\em SIAM J. Discrete Math.}, 25(4):1589--1599, 2011.

\bibitem{hansen}
P.~C. Hansen.
\newblock {\em Discrete Inverse Problems: Insights and Algorithms}.
\newblock SIAM, Philadelphia, 2010.

\bibitem{kubaherman1}
G.~T. Herman and A.~Kuba{ (eds.)}.
\newblock {\em Discrete Tomography: Foundations, Algorithms and Applications}.
\newblock Birkh\"{a}user, Boston, 1999.

\bibitem{kubaherman2}
G.~T. Herman and A.~Kuba{ (eds.)}.
\newblock {\em Advances in Discrete Tomography and its Applications}.
\newblock Birkh\"{a}user, Boston, 2007.

\bibitem{breaklimit1}
R.~Hovden, P.~Ercius, Y.~Jiang, D.~Wang, Y.~Yu, H.~D. Abru{\~n}a, V.~Elser, and
  D.~A. Muller.
\newblock Breaking the {C}rowther limit: {C}ombining depth-sectioning and tilt
  tomography for high-resolution, wide-field {3D} reconstructions.
\newblock {\em Ultramicroscopy}, 140(1):26--31, 2014.

\bibitem{irvingjerrum94}
R.~W. Irving and M.~R. Jerrum.
\newblock Three-dimensional statistical data security problems.
\newblock {\em SIAM J. Comput.}, 23(1):170--184, 1994.

\bibitem{superresolutionCT2}
J.~A. Kennedy, O.~Israel, A.~Frenkel, R.~Bar-Shalom, and H.~Azhari.
\newblock Super-resolution in {PET} imaging.
\newblock {\em IEEE Trans. Med. Imaging}, 25(2):137--147, 2006.

\bibitem{ksbsko-95}
C.~Kisielowski, P.~Schwander, F.~H. Baumann, M.~Seibt, Y.~Kim, and A.~Ourmazd.
\newblock An approach to quantitative high-resolution transmission electron
  microscopy of crystalline materials.
\newblock {\em Ultramicroscopy}, 58(2):131--155, 1995.

\bibitem{kongherman98}
T.~Y. Kong and G.~T. Herman.
\newblock On which grids can tomographic equivalence of binary pictures be
  characterized in terms of elementary switching operations?
\newblock {\em Int. J. Imaging Syst. Technol.}, 9(2-3):118--125, 1998.

\bibitem{kahkrp-09}
A.~K. Kulshreshth, A.~Alpers, G.~T. Herman, E.~Knudsen, L.~Rodek, and H.~F.
  Poulsen.
\newblock A greedy method for reconstructing polycrystals from
  three-dimensional x-ray diffraction data.
\newblock {\em Inverse Probl. Imaging}, 3(1):69--85, 2009.

\bibitem{siam}
E.~L{\'o}pez-Rubio.
\newblock Superresolution from a single noisy image by the median filter
  transform.
\newblock {\em SIAM J. Imaging Sciences}, 9(1):82--115, 2016.

\bibitem{modal7}
J.~D. Malone, M.~T. El-Haddad, I.~Bozic, L.~A. Tye, L.~Majeau, N.~Godbout,
  A.~M. Rollins, C.~Boudoux, K.~M. Joos, S.~N. Patel, and Y.~K. Tao.
\newblock Simultaneous multimodal ophthalmic imaging using swept-source
  spectrally encoded scanning laser ophthalmoscopy and optical coherence
  tomography.
\newblock {\em Biomed. Opt. Express}, 8(1):193--206, 2017.

\bibitem{superresolutionCT3}
G.~E. Marai, D.~H. Laidlaw, and J.~J. Crisco.
\newblock Super-resolution registration using tissue-classified distance
  fields.
\newblock {\em IEEE Trans. Med. Imaging}, 25(2):177--187, 2006.

\bibitem{modal3}
L.~Mart{\'i}-Bonmat{\'i}, R.~Sopena, P.~Bartumeus, and P.~Sopena.
\newblock Multimodality imaging techniques.
\newblock {\em Contrast Media Mol. Imaging}, 5(4):180--189, 2010.

\bibitem{superresolutionbook}
P.~Milanfar{ (ed.)}.
\newblock {\em Super-Resolution Imaging}.
\newblock CRC Press, Boca Raton, 2010.

\bibitem{siltanen}
J.~L. Mueller and S.~Siltanen.
\newblock {\em Linear and Nonlinear Inverse Problems with Practical
  Applications}.
\newblock SIAM, Philadelphia, 2012.

\bibitem{modal5}
A.~D. Parsons, S.~W.~T. Price, N.~Wadeson, M.~Basham, A.~M. Beale, A.~W.
  Ashton, J.~F.~W. Mosselmans, and P.~D. Quinn.
\newblock Automatic processing of multimodal tomography datasets.
\newblock {\em J. Synchrotron Rad.}, 24(1):248--256, 2017.

\bibitem{rpkh07}
L.~Rodek, H.~F. Poulsen, E.~Knudsen, and G.~T. Herman.
\newblock A stochastic algorithm for reconstruction of grain maps of moderately
  deformed specimens based on {X}-ray diffraction.
\newblock {\em J. Appl. Crystallogr.}, 40(2):313--321, 2007.

\bibitem{ryser57}
H.~J. Ryser.
\newblock Combinatorial properties of matrices of zeros and ones.
\newblock {\em Canad. J. Math.}, 9(1):371--377, 1957.

\bibitem{superresolution1}
A.~Schatzberg and A.~J. Devaney.
\newblock Super-resolution in diffraction tomography.
\newblock {\em Inverse Probl.}, 8(1):149--164, 1992.

\bibitem{modal2}
M.~J. Schrapp and G.~T. Herman.
\newblock Data fusion in x-ray computed tomography using a superiorization
  approach.
\newblock {\em Rev. Sci. Instrum.}, 85(5):053701, 2014.

\bibitem{schrijver}
A.~Schrijver.
\newblock {\em Theory of Linear and Integer Programming}.
\newblock John Wiley \& Sons, Chichester, UK, 1986.

\bibitem{sksbko-93}
P.~Schwander, C.~Kisielowski, F.~H. Baumann, Y.~Kim, and A.~Ourmazd.
\newblock Mapping projected potential, interfacial roughness, and composition
  in general crystalline solids by quantitative transmission electron
  microscopy.
\newblock {\em Phys. Rev. Lett.}, 71(25):4150--4153, 1993.

\bibitem{DTnature2}
M.~C. Scott, C.-C. Chen, M.~Mecklenburg, C.~Zhu, R.~Xu, P.~Ercius, U.~Dahmen,
  B.~C. Regan, and J.~Miao.
\newblock Electron tomography at 2.4-{\aa}ngstr{\"o}m resolution.
\newblock {\em Nature}, 483(7390):444--447, 2012.

\bibitem{superresolutionCT4}
E.~Y. Sidky, Y.~Duchin, X.~Pan, and C.~Ullberg.
\newblock A constrained, total-variation minimization algorithm for
  low-intensity x-ray {CT}.
\newblock {\em Med. Phys.}, 38(7):S117, 2011.

\bibitem{superresolutionCT1}
W.~van Aarle, K.~J. Batenburg, G.~Van Gompel, E.~Van de~Casteele, and
  J.~Sijbers.
\newblock Super-resolution for computed tomography based on discrete
  tomography.
\newblock {\em IEEE Trans. Image Process.}, 23(3):1181--1193, 2014.

\bibitem{glidingarc-15}
J.~Zhu, J.~Gao, A.~Ehn, M.~Ald{\'e}n, Z.~Li, D.~Moseev, Y.~Kusano, M.~Salewski,
  A.~Alpers, P.~Gritzmann, and M.~Schwenk.
\newblock Measurements of {3D} slip velocities and plasma column lengths of a
  gliding arc discharge.
\newblock {\em Appl. Phys. Lett.}, 106(4):044101--1--4, 2015.

\end{thebibliography}
\end{document}